\documentclass{article}
\usepackage[english]{babel}
\usepackage{tikz,tikz-3dplot}
\usetikzlibrary{matrix}
\usepackage{amssymb}
\usepackage[T1]{fontenc}
\usepackage[utf8]{inputenc}
\usepackage{calrsfs}
\usepackage{amsmath}
\usepackage{amsthm}
\usepackage{enumerate}
\usepackage{caption}
\usepackage{a4wide}
\usepackage{authblk}
\usetikzlibrary[shapes.misc]
\usetikzlibrary[shapes.geometric]

\newtheorem{theorem}{Theorem}[section]
\newtheorem{cor}[theorem]{Corollary}
\newtheorem{lemma}[theorem]{Lemma}
\newtheorem{prop}[theorem]{Proposition}

\date{\today}

\title{Coloring of the $d^{\text{th}}$ power of the face-centered cubic grid}
\author[1]{Nicolas Gastineau}
\author[1]{Olivier Togni}
\affil[1]{LE2I FRE2005, CNRS, Arts et Métiers, Université Bourgogne Franche-Comté, F-21000 Dijon, France}
\begin{document}

\maketitle

\begin{abstract}
The face-centered cubic grid is a three dimensional 12-regular infinite grid. This graph represents an optimal way to pack spheres in the three-dimensional space. In this grid, the vertices represent the spheres and the edges represent the contact between spheres. We give lower and upper bounds on the chromatic number of the $d^{\text{th}}$ power of the face-centered cubic grid. In particular, in the case $d=2$ we prove that the chromatic number of this grid is 13. We also determine sharper bounds for $d=3$ and for subgraphs of of the face-centered cubic grid.
\end{abstract}

\section{Introduction}
For a graph $G$, we denote by $V(G)$ the \textit{vertex} set of $G$ and by $E(G)\subseteq V(G)\times V(G)$ the \textit{edge set} of $G$.
We denote by $d_{G}(u,v)$, the usual distance between the vertices $u$ and $v$ in a graph $G$. The \emph{diameter} of a graph $G$, denoted by $\text{diam}(G)$ is the maximum distance between the vertices $u$ and $v$ in $G$. By $G[A]$, we denote the graph induced by the set of vertices $A\subseteq V(G)$, i.e., the graph with vertex set $A$ and edge set $\{uv\in E(G)|\ u\in A, v\in A\}$.

A \emph{$k$-coloring} of a graph $G$ is a map $c$ from $V(G)$ to $\{0,1,\ldots,k-1\}$ which satisfies $c(u)\neq c(v)$ for every $uv\in E(G)$.
The \emph{chromatic number} $\chi(G)$ of $G$ is the smallest integer $k$ such that there exists a $k$-coloring of $G$.
The \emph{$d^{\text{th}}$} power $G^{d}$ of a graph $G$ is the graph obtained from $G$ by adding an edge between every two vertices satisfying $d_{G}(u,v)\le d$.

We recall that the survey from Kramer and Kramer~\cite{KR2008} regroup together a good quantity of known results about the coloring of the $d^{\text{th}}$ power of graphs. Note that a coloring of the $d^{\text{th}}$ power of graph is also called a $d$-distance coloring.

Our goal in this paper is to determine the chromatic number of powers of a grid called the face-centered cubic grid. The results of this paper are summarized in Table~\ref{table1}. These results are in the continuity of previous works \cite{FE2003,JJ2005,SE2001} about the coloring of the $d^{\text{th}}$ power of the square, triangular and hexagonal grids.

\begin{table}
\begin{center}
\renewcommand{\arraystretch}{1.5}
\begin{tabular}{|c|c|c|c|c|c|c|c|} 
   \hline
    $k$ & 1 & 2 & 3 & 4 & 5 & odd & even \\
    \hline
    Lower bound on $\chi(\mathcal{F}^d)$ & 4 & 13 & 29 & 55 & 92 & $\frac{5}{12} d^{3}+\frac{5}{4} d^{2}+\frac{19}{12}d+\frac{3}{4}$  & $\frac{5}{12}d^{3}+\frac{5}{4}d^{2}+\frac{11}{6} d+1$ \\
% &  &  &  & & & $+19d /12+3/4$  & $+11/6 d+1$ \\
    \hline 
    Upper bound on $\chi(\mathcal{F}^d)$   & 4 & 13 & 30 & 65 & 108 & \multicolumn{2}{c|}{$(d+1)\lceil (d+1)^2/2 \rceil$}  \\
   \hline 
    Value of $\chi(\mathcal{F}_{0,1}^d)$   & 4 & 9 & 16 & 25 & 36 & \multicolumn{2}{c|}{$(d+1)^2$}  \\
   \hline 
    Lower bound on $\chi(\mathcal{F}_{0,2}^d)$ & 4 & 13 & 24 & 34 & 52 & \multicolumn{2}{c|}{$3 \lceil (d+1)^2/2 \rceil-2$}  \\
   \hline 
    Upper bound on $\chi(\mathcal{F}_{0,2}^d)$ & 4 & 13 & 24 & 36 & 54 & \multicolumn{2}{c|}{$3 \lceil (d+1)^2/2 \rceil$}  \\
    \hline
\end{tabular}
\caption{Our results about the chromatic number of the $d^{\text{th}}$ power of the face-centered cubic grid.}
\label{table1}
\renewcommand{\arraystretch}{1}
\end{center}
\end{table}

The face-centered cubic grid, denoted by $\mathcal{F}$, is the graph with vertex set $\{(i,j,k)|\ i\in \mathbb{Z},\ j\in \mathbb{Z},\ k\in\mathbb{Z}_{2} \}\cup\{(i+0.5,j+0.5,k)|\ i\in \mathbb{Z},\ j\in \mathbb{Z},\ k\in\mathbb{Z}_{1} \}$, where $\mathbb{Z}_{2}$ is the set of the even integers and $\mathbb{Z}_{1}=\mathbb{Z}\setminus\mathbb{Z}_{2}$ (the set of odd integers), and edge set $\{(i,j,k) (i',j',k')|\ (|i-i'|=1\land j=j'\land k=k' )\lor (i=i'\land |j-j'|=1\land k=k')\lor(|i-i'|=1/2\land|j-j'|=1/2 \land |k-k'|=1) \}$. A subgraph of this grid is illustrated in Figure~\ref{fig:FCC}, showing the neighborhood of a vertex.

The layer $k$ of the grid $\mathcal{F}$ is the subset of vertices  $\{(i,j,k)|\ i\in \mathbb{Z},\ j\in \mathbb{Z} \}$, if $k$ is even or the subset of vertices $\{(i+0.5,j+0.5,k)|\ i\in \mathbb{Z},\ j\in \mathbb{Z} \}$, if $k$ is odd. Note that the graph induced by the vertices of layer $k$ is isomorphic to a square grid, see Figure~\ref{fig:FCC}.

We denote by $\mathcal{F}_k$, the subgraph induced by the vertices of layer $k$. 
The graph $\mathcal{F}_{k_1,k_2}$ is the subgraph of $\mathcal{F}$ induced by the vertices from the layer $i$, for $k_1\le i\le k_2$. It can be easily remarked that $\mathcal{F}_{0,k_2 - k_1}$ is isomorphic to $\mathcal{F}_{k_1,k_2}$ for every two integers $k_1$ and $k_2$.

\begin{figure}
\centering
\tdplotsetmaincoords{82}{120}
\begin{tikzpicture}[tdplot_main_coords,scale=1.8]
      \draw[ultra thin,color=red] (2,2,0) -- (2.5,2.5,1.2) -- (2,3,0);
      \draw[ultra thin,color=red] (3,2,0) -- (2.5,2.5,1.2) -- (3,3,0);
      \draw[ultra thin,color=red] (2,2,2.2) -- (2.5,2.5,1.2) -- (2,3,2.2);
      \draw[ultra thin,color=red] (3,2,2.2) -- (2.5,2.5,1.2) -- (3,3,2.2);
\foreach \i in {0,...,4}
 \foreach \j in {0,...,4}
    {
      \ifthenelse{\i < 4}{\draw[thin,color=blue] (\i,\j,0) -- (\i+1,\j,0);}{} 
      \ifthenelse{\i < 4}{\draw[thin,color=blue] (\i,\j,2.2) -- (\i+1,\j,2.2);}{} 
      \ifthenelse{\i < 4}{\draw[thin,color=blue] (\i+.5,\j+.5,1.2) -- (\i+1.5,\j+.5,1.2);}{} 
      \ifthenelse{\j < 4}{\draw[thin,color=blue] (\i,\j,0) -- (\i,\j+1,0);}{} %node[anchor=north west]{$y$};
      \ifthenelse{\j < 4}{\draw[thin,color=blue] (\i,\j,2.2) -- (\i,\j+1,2.2);}{} %node[anchor=north west]{$y$};
      \ifthenelse{\j < 4}{\draw[thin,color=blue] (\i+.5,\j+.5,1.2) -- (\i+.5,\j+1.5,1.2);}{}
      \node[draw, circle, fill, color=gray, minimum size=1pt,inner sep=.7pt] at (\i,\j,0) (x\i\j){};
      \node[draw, circle, fill, color=gray, minimum size=1pt,inner sep=.7pt] at (\i+.5,\j+.5,1.2) (y\i\j){};
      \node[draw, circle, fill, color=gray, minimum size=1pt,inner sep=.7pt] at (\i,\j,2.2) (z\i\j){};
      %\draw[thin,color=red] (\i+.5,\j+.5,1.5) -- (\i,\j,0) -- (\i-.5,\j+.5,1.5); \draw[thin,color=red](\i+.5,\j-.5,1.5) -- (\i,\j,0) -- (\i-.5,\j-.5,1.5); }{}
    }
  \node[draw, circle, fill, color=red, minimum size=4pt,inner sep=1pt] at (2.5,2.5,1.2) (){};

\end{tikzpicture}
\caption{Neighborhood in the face-centered cubic grid $\mathcal{F}$}
\label{fig:FCC}
\end{figure}
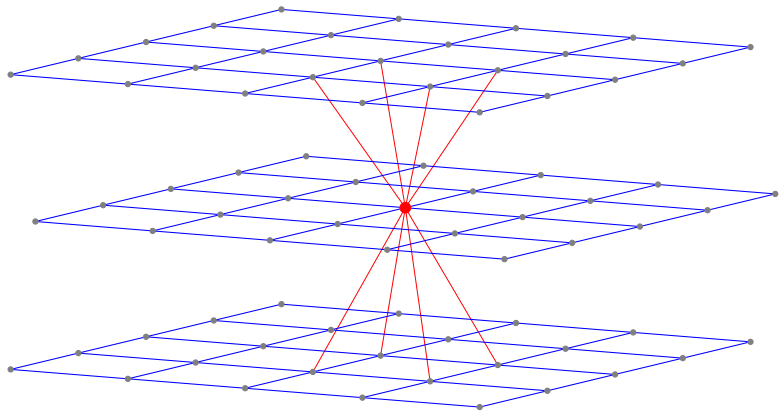

The face-centered cubic grid is a grid studied in the field of programmable matter \cite{BOUR2015,JJD2017,BOUR2018}.
Programmable matter can be seen as modular robots (called modules or particles) able to fix to adjacent modules and send (receive) messages to (from) other modules fixed to the entity. Thus, the different modules form a geometric shape which is a subgraph of a grid. This grid, in the three dimensional case, is almost every time the face-centered cubic grid \cite{BOUR2018}. Note that the face-centered cubic grid is also called the cannonball grid \cite{SP2016}.

It is easily seen that $\chi(\mathcal{F})=4$ since $\mathcal{F}$ contains $K_4$ as subgraph and a $4$-coloring can be obtained by using colors $1$ and $2$ on even layers and colors $3$ and $4$ on odd layers. Note that multicolorings of $\mathcal{F}$ were studied by Šparl et al. \cite{SP2016}.

This paper is organized as follows. 
In Section 2, we determine a formulae for the distance between vertices in the face-centered cubic grid. 
This formulae is used through the remaining sections of this paper.
In Section 3, we determine general lower and upper bounds on $\chi(\mathcal{F}^d)$. 
The lower bound corresponds to the size of a largest complete graph in $\mathcal{F}^d$.
In Section 4, we determine the exact value of $\chi(\mathcal{F}^3)$ and prove that $29\le\chi(\mathcal{F}^4)\le 30$. Section 5 is dedicated to the study of $\chi(\mathcal{F}_{0,1}^d)$ and $\chi(\mathcal{F}_{0,2}^d)$.

\section{Distance between two vertices in the graph $\mathcal{F}$}

Let $p_{+}$ be the function such that $p_{+}(k)=k$ if $k\ge 0$ and $p_{+}(k)=0$ otherwise.

We begin by presenting a formulae that gives the distance between any two vertices of $\mathcal{F}$. 
Note that this formulae restricted to the vertices of a same layer of $\mathcal{F}$ corresponds to the distance in the square grid. 
This proposition will be used through the remaining sections of this paper to calculate the distances in $\mathcal{F}$. 

\begin{prop}\label{prodist}
 The distance between two vertices $(i,j,k)$ and $(i',j',k')$ of $\mathcal{F}$ is given by the following formulae :
$$d_{\mathcal{F}}((i,j,k),(i',j',k'))=p_{+}\left(|i-i'|-\frac{|k-k'|}{2}\right)+p_{+}\left(|j-j'|-\frac{|k-k'|}{2}\right)+ |k-k'|.$$
\end{prop}
\begin{proof}
Note that the formulae can reformulated as follows:
$$d_{\mathcal{F}}((i,j,k),(i',j',k'))= \left\{
    \begin{array}{ll}
        |i-i'|+|j-j'|  & \mbox{if } |i-i'|> |k-k'|/2\land |j-j'|> |k-k'|/2; \\
        |i-i'|+|k-k'|/2  & \mbox{if } |i-i'|> |k-k'|/2\land |j-j'|\le |k-k'|/2; \\
        |j-j'|+|k-k'|/2  & \mbox{if } |i-i'|\le |k-k'|/2\land |j-j'|> |k-k'|/2; \\
        |k-k'|  & \mbox{if } |i-i'|\le |k-k'|/2\land |j-j'|\le |k-k'|/2. \\
    \end{array}
\right.
$$

In this proof, our goal is to show that $d_{\mathcal{F}}((i,j,k),(i',j',k'))=d$, if and only if $p_{+}(|i-i'|-|k-k'|/2)+p_{+}(|j-j'|-|k-k'|/2)+ |k-k'|=d$. We proceed by induction on $d$.
For $d=0$, we easily obtain that $d_{\mathcal{F}}((i,j,k),(i',j',k'))=0$ if and only if $p_{+}(|i-i'|-|k-k'|/2)+p_{+}(|j-j'|-|k-k'|/2)+ |k-k'|=0$. Now suppose that the vertices $(i',j',k')$ at distance $d'$ of $(i,j,k)$ satisfy $p_{+}(|i-i'|-|k-k'|/2)+p_{+}(|j-j'|-|k-k'|/2)+ |k-k''|=d'$, for $d'\le d$. Without loss of generality, suppose $i'=0$, $j'=0$ and $k'=0$. Also, without loss of generality, we can suppose that $i\ge 0$, $j\ge0$ and $k\ge0$ (by symmetry the other cases are proven in the same way).
Thus, the following is true:
$$d_{\mathcal{F}}((0,0,0),(i,j,k))\le d\Leftrightarrow \left\{
    \begin{array}{ll}
        i+j\le d  & \mbox{if } i> k/2\land j> k/2; \\
        i+k/2\le d  & \mbox{if } i> k/2\land j\le k/2; \\
        j+k/2\le d  & \mbox{if } i\le k/2\land j> k/2; \\
        k\le d  & \mbox{if } i\le k/2\land j\le k/2. \\
    \end{array}
\right.
$$
Now, we are going to prove that $d_{\mathcal{F}}((0,0,0),(i,j,k))=d+1$ if and only if $p_{+}(i-k/2)+p_{+}(j-k/2)+ k=d+1$. We cut the proof in four cases: first we consider the vertices such that $i> k/2$ and $j> k/2$, second we consider the vertices such that $i> k/2\land j\le k/2$, third we consider the vertices such that $i\le k/2\land j> k/2$ and finally we consider the vertices which satisfy $i\le k/2$ $j\le k/2$. However, before doing this case analysis proof, we claim the following:

\textbf{Claim 1.} For every integers $i$, $j$, $k$ and $\ell>0$, the following hold:
\begin{enumerate}
\item[i)] $i+j\ge \ell$, implies $d_{\mathcal{F}}((0,0,0),(i,j,k))\ge \ell$;
\item[ii)] $i+k/2\ge \ell$ implies $d_{\mathcal{F}}((0,0,0),(i,j,k)\ge \ell$.
\end{enumerate}
Note that Claim 1.i) is obtained by observing that every two adjacent vertices $(i_1,j_1,k_1)$ and $(i_2,j_2,k_2)$ satisfy $|i_1-i_2|+|j_1-j_2|\le 1$.
Note that every path beginning by $(0,0,0)$ going to $(i,j,k)$ should contain at least $k$ edges between vertices of different layers. 
Also, every two adjacent vertices $(i_1,j_1,k_1)$ and $(i_2,j_2,k_2)$, with $|k_1-k_2|=1$, are such that $|i_1-i_2|=1/2$. Consequently, $i+k/2\ge \ell$ implies $d_{\mathcal{F}}((0,0,0),(i,j,k)\ge i-k/2+k\ge \ell$ and Claim 1.ii) follows.

\begin{description}
\item[Case 1: $i>k/2$ and $j>k/2$.]

First, if $i+j\le d$, then, by induction hypothesis, we have $d_{\mathcal{F}}((0,0,0),(i,j,k))\le d$. Moreover, by Claim 1.i), we obtain that $i+j\ge d+2$ implies $d_{\mathcal{F}}((0,0,0),(i,j,k))\ge d+2$. Consequently, $d_{\mathcal{F}}((0,0,0),(i,j,k))= d+1$ implies $i+j=d+1$.

Second, suppose that $i+j=d+1$. Note that the vertex $( k/2, k/2,k)$ is at distance $k$ from $(0,0,0)$. 
By using the distance in a square grid (the vertices of layer $k$ induce a square grid in $\mathcal{F}$), the vertex $( k/2,  k/2, k)$ is at distance at most $i-k/2+j-k/2=d+1-k$ from $(i,j,k)$. 
Thus, $d_{\mathcal{F}}((0,0,0),(i,j,k))\le d+1$. 
Moreover, by Claim 1.i), we obtain that $i+j\ge d+1$ implies $d_{\mathcal{F}}((0,0,0),(i,j,k))\ge d+1$. 
Finally, we have that $i+j=d+1$ implies $d_{\mathcal{F}}((0,0,0),(i,j,k))= d+1$.

\item[Case 2 : $i>k/2$ and $j\le k/2$.]

First, if $i+k/2\le d$, then, by induction hypothesis $d_{\mathcal{F}}((0,0,0),(i,j,k))\le d$. Moreover, by Claim 1.ii), we obtain that $i+k/2\ge d+2$ implies $d_{\mathcal{F}}((0,0,0),(i,j,k))\ge d+2$. Consequently, $d_{\mathcal{F}}((0,0,0),(i,j,k))= d+1$ implies $i+k/2=d+1$.

Second, suppose that $i+k/2=d+1$. We can easily notice that $d_{\mathcal{F}}((0,0,0),(k/2,j,k))\le k$.
By using the distance in a square grid (the vertices of layer $k$ induce a square grid in $\mathcal{F}$), we have $d_{\mathcal{F}}((k/2,j,k),(i,j,k)\le i-k/2$. Thus, $d_{\mathcal{F}}((0,0,0),(i,j,k))\le k+i-k/2=i+k/2$. Moreover, by Claim 1.ii), we have that $i+k/2\ge d+1$ implies $d_{\mathcal{F}}((0,0,0),(i,j,k))\ge d+1$. Thus, we obtain that $d_{\mathcal{F}}((0,0,0),(i,j,k))=d+1$.

\item[Case 3 : $i\le k/2$ and $j> k/2$.]

The proof is the same than in Case 2 by interchanging the role of $i$ and $j$.

\item[Case 4 : $i\le k/2$ and $j\le k/2$.]

First, if $k\le d$, then, by induction hypothesis, $d_{\mathcal{F}}((0,0,0)(i,j,k))\le d$. Moreover, it can be easily noticed that $k> d+1$ implies $d_{\mathcal{F}}((0,0,0)(i,j,k))> d+1$. Consequently, we obtain that $d_{\mathcal{F}}((0,0,0),(i,j,k))= d+1$ implies $k=d+1$.

Second, suppose that $k=d+1$.
Note that the vertices $(i,j,d+1)$ are at distance at most $d+1$ from $(0,0,0)$. Consequently, we have that $k=d+1$ implies $d_{\mathcal{F}}((0,0,0),(i,j,k))= d+1$.

\end{description}
\end{proof}
\section{Chromatic number of $\mathcal{F}^{d}$}

In this section, we determine general lower and upper bounds on $\chi(\mathcal{F}^d)$. 
Note that there is a gap between the lower and upper bounds. For the majority of the classical grids (square and triangular grids), the value of the chromatic number of $G^d$ corresponds to the size of the largest clique in $G^d$. 
However, that it is not the case for the face-centered cubic grid (see Section 4) and that explains this gap between our lower and upper bounds.

\subsection{Lower bound on the chromatic number of $\mathcal{F}^{d}$}
The following Proposition is a consequence of a result of Bjornholm~\cite{BJO1990} which has proven that the set $A_{\ell }=\{(i,j,k)\in V(\mathcal{F})|\ d_{\mathcal{F}}((0,0,0)(i,j,k))\le \ell \}$ is such that $|A_{\ell }|=(2 \ell +1)(5\ell ^2+5\ell +3)/3$ (On-line encyclopedia of integer sequences number: A005902). Since $\text{diam}(G[A_{\ell}])=2\ell$ and by setting $d=\ell /2$, we obtain the following result.
\begin{prop}[\cite{BJO1990}]\label{evenbound}
If $d$ is even, then $\chi(\mathcal{F}^{d})\ge \frac{5}{12} d^{3}+\frac{5}{4} d^{2}+\frac{11}{6} d +1$.
\end{prop}
Let $D_{0}=\{(0,0,0),(1,0,0),(0.5,0.5,1),(0.5,-0.5,1)\}$. Note that $D_{0}$ induces a complete graph of four vertices in the grid $\mathcal{F}$. Let $D_{\ell }=\{ u\in\mathcal{F}|\ \min_{v\in D_{0}}(d_{\mathcal{F}}(u,v))=\ell \}$.

In the two following Lemmas, we calculate the size of a largest clique in $\mathcal{F}^{d}$ when $d$ is odd. From this value, we will infer a lower bound on $\chi(\mathcal{F}^d)$ (as in Proposition \ref{evenbound}) for an odd $d$.
\begin{lemma}\label{le1}
For any $\ell\ge 1$, $|D_{\ell }|=10 \ell^{2}+10\ell +4$.
\end{lemma}
\begin{proof}
Note that we have $|D_{0}|=4$ and that $D_1=\{(-0.5,0.5,-1),$ $(-0.5,-0.5,-1), $ $(0.5,0.5,-1),$ $(0.5,-0.5,-1),$ $(1.5,0.5,-1),$ $(1.5,-0.5,-1)\}\cup\{(0,1,0),$ $(0,-1,0),$ $(1,1,0),$ $(1,-1,0),$ $(-1,0,0),$ $(2,0,0)\}\cup\{(-0.5,0.5,1),$ $(1.5,0.5,1),$ $(-0.5,-0.5,1),$ $(1.5,-0.5,1),$ $(0.5,1.5,1),$ $(0.5,-1.5,1)\}\cup\{ (0,-1,2),$ $(1,-1,2),$ $(0,0,2),$ $(1,0,2),$ $(0,1,2),$ $(1,1,2) \}$. 
By induction on $\ell$, we will prove that there are $(\ell+1) (\ell+2)$ vertices in layers $-\ell$ and $\ell+1$ in $D_{\ell}$ and that there are $2+4\ell$ vertices in layer $j$ in $D_{\ell}$, for $-\ell<j< \ell+1 $.

If $\ell =1$, we can easily remark that there are $6$ vertices of $D_{1}$ in each layer $i$, $i\in \{-1,0,1,2\}$. 
Thus, $|D_1|=24$.
Now, suppose by induction that there are $\ell (\ell+1)$ vertices in layers $-(\ell -1)$ and $\ell$ in $D_{\ell-1 }$ and that there are $2+4(\ell -1)$ vertices in the layer $i$ in $D_{\ell-1}$, for $-(\ell-1)<i< \ell $.

It can be easily verified that there are $(\ell +2)(\ell +1)$ vertices in layers $-\ell $ and $\ell+1$ in $D_{\ell}$. 
Now, since we have supposed that there are $2+4(\ell -1)$ vertices in layer $i$ in $D_{\ell-1 }$, for $-(\ell -1)<i< \ell $, we obtain that there are $2+4\ell $ vertices in layer $i$ in $D_{\ell }$, for $-(\ell-1)<i<\ell $ (since there are four more vertices in each layer in $D_{\ell }$). 
Also, it can be easily verified that there are $2+4\ell $ vertices in layers $-(\ell -1)$ and $\ell$ in $D_{\ell }$ and the property follows.

By calculation, we obtain: $$|D_{\ell }|=2(\ell+1) (\ell +2)+ 2\ell (2+4 \ell)=2\ell^2+6\ell+4+8\ell^2+4\ell =10\ell^{2}+10\ell+4.$$
\end{proof}
For an integer $\ell \ge 0$, let $B_{\ell }=\cup_{0\le i\le \ell } D_{i}$.
\begin{lemma}\label{le2}
For any integer $\ell \ge 0$, $\text{diam}(\mathcal{F}[B_{\ell}])=2\ell+1$.
\end{lemma}
\begin{proof}
Note that $\mathcal{F}[B_{0}]$ is an induced complete graph of four vertices and thus $\text{diam}(\mathcal{F}[B_{0}])=1$. By construction, we have $\text{diam}(\mathcal{F}[B_{\ell+1}])=diam(\mathcal{F}[B_{\ell }])+2$. Thus, by induction, we have $\text{diam}(\mathcal{F}[B_{\ell }])=2\ell +1$.
\end{proof}
As in Proposition \ref{evenbound}, we infer the following lower bound on $\chi(\mathcal{F}^d)$, for an odd $d$.
\begin{prop}\label{oddbound}
If $d$ is odd, then $$\chi(\mathcal{F}^{d})\ge \frac{5}{12} d^{3}+\frac{5}{4} d^{2}+\frac{19}{12} d +\frac{3}{4}.$$
\end{prop}
\begin{proof}
Since, by Lemma \ref{le2}, $\text{diam}(\mathcal{F}[B_{\ell}])=2\ell +1$, we have $\chi(\mathcal{F}^{2 \ell+1})\ge |B_{\ell}|$. By Lemma \ref{le1}, $|D_{\ell}|=10 \ell^{2}+10\ell+4$ and by calculation: 
\begin{center}
   $$
   \begin{array}{rcl}
|B_{\ell}|=\sum_{i=0}^{\ell} |D_{i}|=\sum_{i=0}^{\ell} (10i^{2}+10i+4)=10 \sum_{i=0}^{\ell }i^{2}+10 \sum_{i=0}^{\ell } i+4\ell\\ \\
= 10(\frac{1}{6} \ell (2\ell^{2}+3\ell +1))+10 (\frac{1}{2} \ell(\ell+1))+4(\ell+1)=\frac{10}{3} \ell^{3}+10\ell^2 +\frac{32}{3} \ell+4.
   \end{array}
   $$
\end{center}
Also, by calculation: $$|B_{(d-1)/2}|=\frac{5}{12} d^{3}+\frac{5}{4} d^{2}+\frac{19}{12} d +\frac{3}{4}.$$ 
Thus, if $d$ is odd, then we obtain that $$\chi(\mathcal{F}^{d})\ge|B_{(d-1)/2}|=\frac{5}{12} d^{3}+\frac{5}{4} d^{2}+\frac{19}{12} d +\frac{3}{4}.$$
\end{proof}
\subsection{Upper bound on the chromatic number of $\mathcal{F}^{d}$}
We determine the following upper bound on the chromatic number of $\mathcal{F}^{d}$.
Note that this upper bound is reached for $d=1$ and is not reached for $d=2$ or $d=3$ (see Section 4).
\begin{theorem}\label{upbound}
For any $d\ge 1$, $\chi(\mathcal{F}^{d})\le (d+1) \lceil (d+1)^{2}/2 \rceil.$
\end{theorem}
\begin{proof}
By Proposition \ref{prodist}, the distance between two vertices $(i,j,k)$ and $(i',j',k)$ from the same layer is $|i-i|+|j-j'|$.
This distance corresponds to the distance between two vertices $(i,j)$ and $(i',j')$ in the square grid.
Moreover, the distance between two vertices $(i,j,k)$ and $(i',j',k')$ is at least $d+1$ if $|k-k'|\ge d+1$.

Therefore, by using $(d+1)$ different patterns of $\lceil (d+1)^{2}/2 \rceil$ colors used for coloring the $d^{\text{th}}$ power of the square grid \cite{FE2003} it is possible to color $\mathcal{F}^{d}$. 
We use the $a^{\text{th}}$ pattern to color the layer $b$, if $b\equiv a\pmod {d+1}$. Since in each pattern we use $\lceil (d+1)^{2}/2 \rceil$ different colors, in total we have used  $(d+1) \lceil (d+1)^{2}/2 \rceil$ colors.
\end{proof}

\section{Chromatic number of the $2^{\text{nd}}$ and the $3^{\text{th}}$ power of $\mathcal{F}$}
\subsection{Chromatic number of the $2^{\text{nd}}$ power of $\mathcal{F}$}
Note that the value of $\chi(\mathcal{F}^{2})$ corresponds to the lower bound of Proposition \ref{evenbound} which itself corresponds to the size of the largest clique in $\mathcal{F}^2$.
\begin{theorem}\label{balancetontheoreme}
$\chi(\mathcal{F}^{2})=13$.
\end{theorem}
\begin{proof}
By Proposition \ref{evenbound}, we have $\chi(\mathcal{F}^2)\ge 40/12+5+11/3+1=13$. Now, it remains to prove that $\chi(\mathcal{F}^{2})\le 13$.
We set the following coloring function:
$$\forall (i,j,k)\in V(\mathcal{F}),\ c((i,j,k)) = i-2j+\frac{9k}{2}\pmod{13}.$$
We claim that $c$ is a coloring function of the $2^{\text{nd}}$ power of $\mathcal{F}$.
Let $u$ be a vertex in $\mathcal{F}$ such that $c(u)=0$. In Figure~\ref{colorpower2}, we represent the graph induced by the vertices at distance at most $2$ from $u$. It is easy to verify that no vertex of this subgraph except $u$ has color $0$.

In order to prove that every two vertices $(i,j,k)$ and $(i',j',k')$ satisfying $c((i,j,k))=c((i',j',k'))$ are such that $d_{\mathcal{F}}((i,j,k),(i',j',k'))>2$, we consider three cases: first the case $k'=k$, second the case $|k-k'|=1$ and finally the case $|k-k'|=2$.

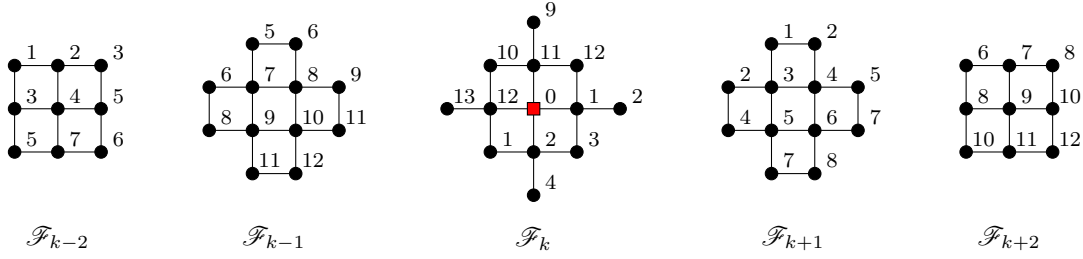
\begin{figure}[t]
\begin{center}
\begin{tikzpicture}[scale=1.15]

\draw (-1,0) -- (1,0);
\draw (-0.5,0.5) -- (0.5,0.5);
\draw (-0.5,-0.5) -- (0.5,-0.5);
\draw (0,-1) -- (0,1);
\draw (0.5,-0.5) -- (0.5,0.5);
\draw (-0.5,-0.5) -- (-0.5,0.5);

\node at (0,0)[regular polygon, regular polygon sides=4,draw=black,fill=red,scale=0.5]{};
\node at (0.5,0)[circle,draw=black,fill=black,scale=0.5]{};
\node at (1,0)[circle,draw=black,fill=black,scale=0.5]{};
\node at (-0.5,0)[circle,draw=black,fill=black,scale=0.5]{};
\node at (-1,0)[circle,draw=black,fill=black,scale=0.5]{};
\node at (0,0.5)[circle,draw=black,fill=black,scale=0.5]{};
\node at (0.5,0.5)[circle,draw=black,fill=black,scale=0.5]{};
\node at (-0.5,0.5)[circle,draw=black,fill=black,scale=0.5]{};
\node at (0,-0.5)[circle,draw=black,fill=black,scale=0.5]{};
\node at (0.5,-0.5)[circle,draw=black,fill=black,scale=0.5]{};
\node at (-0.5,-0.5)[circle,draw=black,fill=black,scale=0.5]{};
\node at (0,1)[circle,draw=black,fill=black,scale=0.5]{};
\node at (0,-1)[circle,draw=black,fill=black,scale=0.5]{};

\node at (0.2,0.15){\footnotesize{$0$}};
\node at (0.7,0.15){\footnotesize{$1$}};
\node at (1.2,0.15){\footnotesize{$2$}};
\node at (-0.3,0.15){\footnotesize{$12$}};
\node at (-0.8,0.15){\footnotesize{$13$}};
\node at (0.2,0.65){\footnotesize{$11$}};
\node at (0.7,0.65){\footnotesize{$12$}};
\node at (-0.3,0.65){\footnotesize{$10$}};
\node at (0.2,-0.35){\footnotesize{$2$}};
\node at (0.7,-0.35){\footnotesize{$3$}};
\node at (-0.3,-0.35){\footnotesize{$1$}};
\node at (0.2,1.15){\footnotesize{$9$}};
\node at (0.2,-0.85){\footnotesize{$4$}};

\node at (0,-1.5){$\mathcal{F}_k$};
\node at (-3,-1.5){$\mathcal{F}_{k-1}$};
\node at (-5.5,-1.5){$\mathcal{F}_{k-2}$};
\node at (3,-1.5){$\mathcal{F}_{k+1}$};
\node at (5.5,-1.5){$\mathcal{F}_{k+2}$};

\draw (-0.75-3,0.25) -- (0.75-3,0.25);
\draw (-0.25-3,0.75) -- (0.25-3,0.75);
\draw (-0.75-3,-0.25) -- (0.75-3,-0.25);
\draw (-0.25-3,-0.75) -- (0.25-3,-0.75);

\draw (0.25-3,-0.75) -- (0.25-3,0.75);
\draw (0.75-3,-0.25) -- (0.75-3,0.25);
\draw (-0.25-3,-0.75) -- (-0.25-3,0.75);
\draw (-0.75-3,-0.25) -- (-0.75-3,0.25);

\node at (0.25-3,0.25)[circle,draw=black,fill=black,scale=0.5]{};
\node at (0.75-3,0.25)[circle,draw=black,fill=black,scale=0.5]{};
\node at (-0.25-3,0.25)[circle,draw=black,fill=black,scale=0.5]{};
\node at (-0.75-3,0.25)[circle,draw=black,fill=black,scale=0.5]{};
\node at (0.25-3,-0.25)[circle,draw=black,fill=black,scale=0.5]{};
\node at (0.75-3,-0.25)[circle,draw=black,fill=black,scale=0.5]{};
\node at (-0.25-3,-0.25)[circle,draw=black,fill=black,scale=0.5]{};
\node at (-0.75-3,-0.25)[circle,draw=black,fill=black,scale=0.5]{};
\node at (0.25-3,0.75)[circle,draw=black,fill=black,scale=0.5]{};
\node at (-0.25-3,0.75)[circle,draw=black,fill=black,scale=0.5]{};
\node at (0.25-3,-0.75)[circle,draw=black,fill=black,scale=0.5]{};
\node at (-0.25-3,-0.75)[circle,draw=black,fill=black,scale=0.5]{};

\node at (0.45-3,0.4){\footnotesize{$8$}};
\node at (0.95-3,0.4){\footnotesize{$9$}};
\node at (-0.05-3,0.4){\footnotesize{$7$}};
\node at (-0.55-3,0.4){\footnotesize{$6$}};
\node at (0.45-3,-0.1){\footnotesize{$10$}};
\node at (0.95-3,-0.1){\footnotesize{$11$}};
\node at (-0.05-3,-0.1){\footnotesize{$9$}};
\node at (-0.55-3,-0.1){\footnotesize{$8$}};
\node at (0.45-3,0.9){\footnotesize{$6$}};
\node at (-0.05-3,0.9){\footnotesize{$5$}};
\node at (0.45-3,-0.6){\footnotesize{$12$}};
\node at (-0.05-3,-0.6){\footnotesize{$11$}};

\draw (-0.5-5.5,0.5) -- (0.5-5.5,0.5);
\draw (-0.5-5.5,0) -- (0.5-5.5,0);
\draw (-0.5-5.5,-0.5) -- (0.5-5.5,-0.5);
\draw (0.5-5.5,-0.5) -- (0.5-5.5,0.5);
\draw (0-5.5,-0.5) -- (0-5.5,0.5);
\draw (-0.5-5.5,-0.5) -- (-0.5-5.5,0.5);

\node at (-0.5-5.5,0.5)[circle,draw=black,fill=black,scale=0.5]{};
\node at (0-5.5,0.5)[circle,draw=black,fill=black,scale=0.5]{};
\node at (0.5-5.5,0.5)[circle,draw=black,fill=black,scale=0.5]{};
\node at (-0.5-5.5,0)[circle,draw=black,fill=black,scale=0.5]{};
\node at (0-5.5,0)[circle,draw=black,fill=black,scale=0.5]{};
\node at (0.5-5.5,0)[circle,draw=black,fill=black,scale=0.5]{};
\node at (-0.5-5.5,-0.5)[circle,draw=black,fill=black,scale=0.5]{};
\node at (0-5.5,-0.5)[circle,draw=black,fill=black,scale=0.5]{};
\node at (0.5-5.5,-0.5)[circle,draw=black,fill=black,scale=0.5]{};

\node at (0.2-5.5,0.15){\footnotesize{$4$}};
\node at (0.7-5.5,0.15){\footnotesize{$5$}};
\node at (-0.3-5.5,0.15){\footnotesize{$3$}};
\node at (0.2-5.5,0.65){\footnotesize{$2$}};
\node at (0.7-5.5,0.65){\footnotesize{$3$}};
\node at (-0.3-5.5,0.65){\footnotesize{$1$}};
\node at (0.2-5.5,-0.35){\footnotesize{$7$}};
\node at (0.7-5.5,-0.35){\footnotesize{$6$}};
\node at (-0.3-5.5,-0.35){\footnotesize{$5$}};

\draw (-0.75+3,0.25) -- (0.75+3,0.25);
\draw (-0.25+3,0.75) -- (0.25+3,0.75);
\draw (-0.75+3,-0.25) -- (0.75+3,-0.25);
\draw (-0.25+3,-0.75) -- (0.25+3,-0.75);

\draw (0.25+3,-0.75) -- (0.25+3,0.75);
\draw (0.75+3,-0.25) -- (0.75+3,0.25);
\draw (-0.25+3,-0.75) -- (-0.25+3,0.75);
\draw (-0.75+3,-0.25) -- (-0.75+3,0.25);

\node at (0.25+3,0.25)[circle,draw=black,fill=black,scale=0.5]{};
\node at (0.75+3,0.25)[circle,draw=black,fill=black,scale=0.5]{};
\node at (-0.25+3,0.25)[circle,draw=black,fill=black,scale=0.5]{};
\node at (-0.75+3,0.25)[circle,draw=black,fill=black,scale=0.5]{};
\node at (0.25+3,-0.25)[circle,draw=black,fill=black,scale=0.5]{};
\node at (0.75+3,-0.25)[circle,draw=black,fill=black,scale=0.5]{};
\node at (-0.25+3,-0.25)[circle,draw=black,fill=black,scale=0.5]{};
\node at (-0.75+3,-0.25)[circle,draw=black,fill=black,scale=0.5]{};
\node at (0.25+3,0.75)[circle,draw=black,fill=black,scale=0.5]{};
\node at (-0.25+3,0.75)[circle,draw=black,fill=black,scale=0.5]{};
\node at (0.25+3,-0.75)[circle,draw=black,fill=black,scale=0.5]{};
\node at (-0.25+3,-0.75)[circle,draw=black,fill=black,scale=0.5]{};

\node at (0.45+3,0.4){\footnotesize{$4$}};
\node at (0.95+3,0.4){\footnotesize{$5$}};
\node at (-0.05+3,0.4){\footnotesize{$3$}};
\node at (-0.55+3,0.4){\footnotesize{$2$}};
\node at (0.45+3,-0.1){\footnotesize{$6$}};
\node at (0.95+3,-0.1){\footnotesize{$7$}};
\node at (-0.05+3,-0.1){\footnotesize{$5$}};
\node at (-0.55+3,-0.1){\footnotesize{$4$}};
\node at (0.45+3,0.9){\footnotesize{$2$}};
\node at (-0.05+3,0.9){\footnotesize{$1$}};
\node at (0.45+3,-0.6){\footnotesize{$8$}};
\node at (-0.05+3,-0.6){\footnotesize{$7$}};

\draw (-0.5+5.5,0.5) -- (0.5+5.5,0.5);
\draw (-0.5+5.5,0) -- (0.5+5.5,0);
\draw (-0.5+5.5,-0.5) -- (0.5+5.5,-0.5);
\draw (0.5+5.5,-0.5) -- (0.5+5.5,0.5);
\draw (0+5.5,-0.5) -- (0+5.5,0.5);
\draw (-0.5+5.5,-0.5) -- (-0.5+5.5,0.5);

\node at (-0.5+5.5,0.5)[circle,draw=black,fill=black,scale=0.5]{};
\node at (0+5.5,0.5)[circle,draw=black,fill=black,scale=0.5]{};
\node at (0.5+5.5,0.5)[circle,draw=black,fill=black,scale=0.5]{};
\node at (-0.5+5.5,0)[circle,draw=black,fill=black,scale=0.5]{};
\node at (0+5.5,0)[circle,draw=black,fill=black,scale=0.5]{};
\node at (0.5+5.5,0)[circle,draw=black,fill=black,scale=0.5]{};
\node at (-0.5+5.5,-0.5)[circle,draw=black,fill=black,scale=0.5]{};
\node at (0+5.5,-0.5)[circle,draw=black,fill=black,scale=0.5]{};
\node at (0.5+5.5,-0.5)[circle,draw=black,fill=black,scale=0.5]{};

\node at (0.2+5.5,0.15){\footnotesize{$9$}};
\node at (0.7+5.5,0.15){\footnotesize{$10$}};
\node at (-0.3+5.5,0.15){\footnotesize{$8$}};
\node at (0.2+5.5,0.65){\footnotesize{$7$}};
\node at (0.7+5.5,0.65){\footnotesize{$8$}};
\node at (-0.3+5.5,0.65){\footnotesize{$6$}};
\node at (0.2+5.5,-0.35){\footnotesize{$11$}};
\node at (0.7+5.5,-0.35){\footnotesize{$12$}};
\node at (-0.3+5.5,-0.35){\footnotesize{$10$}};

\end{tikzpicture}
\end{center}
\caption{The colors of the vertices at distance at most 2 of a fixed vertex $u$ of color $0$ in the layer $k$ of the graph $\mathcal{F}$ (square: $u$; circle: vertex at distance $d$ of $u$, $1\le d\le 2$).}
\label{colorpower2}
\end{figure}

\begin{description}
\item[Case 1: $k'=k$.]
The vertices $(i',j',k)$ at distance at most 2 from $(i,j,k)$ are such that $|i-i'|+ |j-j'|\le 2$. 
First, if $1\le |i-i'|\le 2$ and $j=j'$, then we have $(i-i')\not\equiv 0\pmod{13}$. Consequently, since $c((i,j,k))-c((i',j',k)\equiv (i-i')\pmod{13}$, we have $c((i,j,k))\neq c((i',j',k))$. 
Second, if $1\le |j-j'|\le 2$ and $i=i'$, then we have $2(j'-j)\not\equiv 0\pmod{13}$. Consequently, since $c((i,j,k))-c((i',j',k)\equiv 2(j'-j)\pmod{13}$, we have $c((i,j,k))\neq c((i',j',k))$. 
Finally, if $|i-i'|=1$ and $|j-j'|=1$, then we have $(i-i')+2(j'-j)\not\equiv 0 \pmod{13}$. Consequently, since $c((i,j,k))-c((i',j',k)\equiv (i-i')+2(j'-j)\pmod{13}$ , we have $c((i,j,k))\neq c((i',j',k))$. Thus, no vertex at distance at most 2 from $(i,j,k)$ has the same color than $(i,j,k)$ in the layer $k$.

\item[Case 2: $|k'-k|=1$.]
The vertices $(i',j',k+1)$ ( $(i',j',k-1)$, respectively) at distance at most 2 from $(i,j,k)$ are such that $|i-i'|=1/2\land |j-j'|\le 3/2 $ or such that $|i-i'|\le 3/2 \land |j-j'|=1/2$. Therefore, we have $(i-i')+2(j'-j)\not\equiv 9/2 \pmod{13}$ ( $(i-i')+2(j'-j)\not\equiv -9/2\pmod{13}$, respectively). Thus, since $c((i,j,k))-c((i',j',k+1)\equiv (i-i')+2(j'-j)-9/2\pmod{13}$ ( $c((i,j,k))-c((i',j',k-1)\equiv (i-i')+2(j'-j)+9/2 \pmod{13}$, respectively), we have $c((i,j,k))\neq c((i',j',k+1))$ ($c((i,j,k))\neq c((i',j',k-1))$, respectively). Therefore, no vertex at distance at most 2 from $(i,j,k)$ has the same color than $(i,j,k)$ in the layer $k+1$ ($k-1$, respectively).

\item[Case 3: $|k'-k|=2$.]
The vertices $(i,j,k+2)$ ($(i',j',k-2)$, respectively) at distance at most 2 from $(i,j,k)$ are such that $|i-i'|\le 1\land |j-j'|\le 1$. Therefore, we have $(i-i')+2(j'-j)\not\equiv 9\pmod{13}$ ($(i-i')+2(j'-j)\not\equiv 4\pmod{13}$, respectively). Thus, since $c((i,j,k))-c((i',j',k+2)\equiv (i-i')+2(j'-j)-9\pmod{13}$ ( $c((i,j,k))-c((i',j',k-2)\equiv (i-i')+2(j'-j)+9\pmod{13}$, respectively), we have $c((i,j,k))\neq c((i',j',k+2))$ ($c((i,j,k))\neq c((i',j',k-2))$, respectively). Therefore, no vertex at distance at most 2 from $(i,j,k)$ has the same color than $(i,j,k)$ in the layer $k+2$ ($k-2$, respectively).

\end{description}
Finally, since every vertex of layer $d$, for $d\le k-3$ or $d\ge k+3$ is at distance at least 3 of $(i,j,k)$, we obtain that no vertex at distance at most $2$ from $(i,j,k)$ has the same color than $(i,j,k)$.
\end{proof}
\subsection{Chromatic number of the $3^{\text{th}}$ power of $\mathcal{F}$}
Note that the value of $\chi(\mathcal{F}^{3})$ is smaller than the upper bound of Proposition \ref{upbound} (which is 32, for $d=3$).
In the following theorem, we denote by $i_{(\text{mod } j)}$ or by $i\pmod{j}$ the smallest positive integer $a$ such $i\equiv a\pmod{30}$
\begin{theorem}
$\chi(\mathcal{F}^{3})\le 30$.
\end{theorem}
\begin{proof}
We set the following function: $\forall (i,j,k)\in V(\mathcal{F}),$
$$ c((i,j,k)) = \left\{
    \begin{array}{ll}
        i_{(\text{mod } 5)}+5j+15k/2\pmod{30} & \mbox{if } k\equiv0\pmod{2}; \\
        (i-1/2)_{(\text{mod } 5)}+5(j-1/2)+8 +15(k-1)/2\pmod{30} & \text{otherwise}. \\
    \end{array}
\right.
$$

We claim that $c$ is a coloring function of the $3^{\text{th}}$ power of $\mathcal{F}$.
It remains to prove that every two vertices $(i,j,k)$ and $(i',j',k')$ satisfying $c((i,j,k))=c((i',j',k'))$ are such that $d_{\mathcal{F}}((i,j,k),(i',j',k'))>3$.
This can be proven by a (tedious) case analysis similar with the one of the proof of Theorem ~\ref{balancetontheoreme}. Instead, we only give an illustrative argument.

Let $u$ be a vertex in $\mathcal{F}$. Without loss of generality, suppose $u$ has color $0$. In Figure~\ref{colorpower3}, we represent the graph induced by the vertices at distance at most $3$ from $u$. It is easy to verify that no vertex of this subgraph except $u$ has color $0$.

\begin{figure}[t]
\begin{center}
\begin{tikzpicture}[scale=1.15]

\draw (-1.5,0) -- (1.5,0);
\draw (-1,0.5) -- (1,0.5);
\draw (-0.5,1) -- (0.5,1);
\draw (-1,-0.5) -- (1,-0.5);
\draw (-0.5,-1) -- (0.5,-1);
\draw (0,-1.5) -- (0,1.5);
\draw (0.5,-1) -- (0.5,1);
\draw (1,-0.5) -- (1,0.5);
\draw (-0.5,-1) -- (-0.5,1);
\draw (-1,-0.5) -- (-1,0.5);

\node at (0,0)[regular polygon, regular polygon sides=4,draw=black,fill=red,scale=0.5]{};
\node at (0.5,0)[circle,draw=black,fill=black,scale=0.5]{};
\node at (1,0)[circle,draw=black,fill=black,scale=0.5]{};
\node at (1.5,0)[circle,draw=black,fill=black,scale=0.5]{};
\node at (-0.5,0)[circle,draw=black,fill=black,scale=0.5]{};
\node at (-1,0)[circle,draw=black,fill=black,scale=0.5]{};
\node at (-1.5,0)[circle,draw=black,fill=black,scale=0.5]{};
\node at (0,0.5)[circle,draw=black,fill=black,scale=0.5]{};
\node at (0.5,0.5)[circle,draw=black,fill=black,scale=0.5]{};
\node at (1,0.5)[circle,draw=black,fill=black,scale=0.5]{};
\node at (-0.5,0.5)[circle,draw=black,fill=black,scale=0.5]{};
\node at (-1,0.5)[circle,draw=black,fill=black,scale=0.5]{};
\node at (0,-0.5)[circle,draw=black,fill=black,scale=0.5]{};
\node at (0.5,-0.5)[circle,draw=black,fill=black,scale=0.5]{};
\node at (1,-0.5)[circle,draw=black,fill=black,scale=0.5]{};
\node at (-0.5,-0.5)[circle,draw=black,fill=black,scale=0.5]{};
\node at (-1,-0.5)[circle,draw=black,fill=black,scale=0.5]{};
\node at (0,1)[circle,draw=black,fill=black,scale=0.5]{};
\node at (0.5,1)[circle,draw=black,fill=black,scale=0.5]{};
\node at (-0.5,1)[circle,draw=black,fill=black,scale=0.5]{};
\node at (0,-1)[circle,draw=black,fill=black,scale=0.5]{};
\node at (0.5,-1)[circle,draw=black,fill=black,scale=0.5]{};
\node at (-0.5,-1)[circle,draw=black,fill=black,scale=0.5]{};
\node at (0,1.5)[circle,draw=black,fill=black,scale=0.5]{};
\node at (0,-1.5)[circle,draw=black,fill=black,scale=0.5]{};

\node at (0.2,0.15){\footnotesize{$0$}};
\node at (0.7,0.15){\footnotesize{$1$}};
\node at (1.2,0.15){\footnotesize{$2$}};
\node at (1.7,0.15){\footnotesize{$3$}};
\node at (-0.3,0.15){\footnotesize{$4$}};
\node at (-0.8,0.15){\footnotesize{$3$}};
\node at (-1.3,0.15){\footnotesize{$2$}};
\node at (0.2,0.65){\footnotesize{$25$}};
\node at (0.7,0.65){\footnotesize{$26$}};
\node at (1.2,0.65){\footnotesize{$27$}};
\node at (-0.3,0.65){\footnotesize{$29$}};
\node at (-0.8,0.65){\footnotesize{$28$}};
\node at (0.2,-0.35){\footnotesize{$5$}};
\node at (0.7,-0.35){\footnotesize{$6$}};
\node at (1.2,-0.35){\footnotesize{$7$}};
\node at (-0.3,-0.35){\footnotesize{$9$}};
\node at (-0.8,-0.35){\footnotesize{$8$}};
\node at (0.2,1.15){\footnotesize{$20$}};
\node at (0.7,1.15){\footnotesize{$21$}};
\node at (-0.3,1.15){\footnotesize{$24$}};
\node at (0.2,-0.85){\footnotesize{$10$}};
\node at (0.7,-0.85){\footnotesize{$11$}};
\node at (-0.3,-0.85){\footnotesize{$14$}};
\node at (0.2,1.65){\footnotesize{$15$}};
\node at (0.2,-1.35){\footnotesize{$15$}};

\node at (0,-2){$\mathcal{F}_k$};
\node at (-3.5,-2){$\mathcal{F}_{k-1}$};
\node at (-7,-2){$\mathcal{F}_{k-2}$};
\node at (-10.5,-2){$\mathcal{F}_{k-3}$};

\draw (-1.25-3.5,0.25) -- (1.25-3.5,0.25);
\draw (-0.75-3.5,0.75) -- (0.75-3.5,0.75);
\draw (-0.25-3.5,1.25) -- (0.25-3.5,1.25);
\draw (-1.25-3.5,-0.25) -- (1.25-3.5,-0.25);
\draw (-0.75-3.5,-0.75) -- (0.75-3.5,-0.75);
\draw (-0.25-3.5,-1.25) -- (0.25-3.5,-1.25);

\draw (0.25-3.5,-1.25) -- (0.25-3.5,1.25);
\draw (0.75-3.5,-0.75) -- (0.75-3.5,0.75);
\draw (1.25-3.5,-0.25) -- (1.25-3.5,0.25);
\draw (-0.25-3.5,-1.25) -- (-0.25-3.5,1.25);
\draw (-0.75-3.5,-0.75) -- (-0.75-3.5,0.75);
\draw (-1.25-3.5,-0.25) -- (-1.25-3.5,0.25);

\node at (0.25-3.5,0.25)[circle,draw=black,fill=black,scale=0.5]{};
\node at (0.75-3.5,0.25)[circle,draw=black,fill=black,scale=0.5]{};
\node at (1.25-3.5,0.25)[circle,draw=black,fill=black,scale=0.5]{};
\node at (-0.25-3.5,0.25)[circle,draw=black,fill=black,scale=0.5]{};
\node at (-0.75-3.5,0.25)[circle,draw=black,fill=black,scale=0.5]{};
\node at (-1.25-3.5,0.25)[circle,draw=black,fill=black,scale=0.5]{};
\node at (0.25-3.5,-0.25)[circle,draw=black,fill=black,scale=0.5]{};
\node at (0.75-3.5,-0.25)[circle,draw=black,fill=black,scale=0.5]{};
\node at (1.25-3.5,-0.25)[circle,draw=black,fill=black,scale=0.5]{};
\node at (-0.25-3.5,-0.25)[circle,draw=black,fill=black,scale=0.5]{};
\node at (-0.75-3.5,-0.25)[circle,draw=black,fill=black,scale=0.5]{};
\node at (-1.25-3.5,-0.25)[circle,draw=black,fill=black,scale=0.5]{};
\node at (0.25-3.5,0.75)[circle,draw=black,fill=black,scale=0.5]{};
\node at (0.75-3.5,0.75)[circle,draw=black,fill=black,scale=0.5]{};
\node at (-0.25-3.5,0.75)[circle,draw=black,fill=black,scale=0.5]{};
\node at (-0.75-3.5,0.75)[circle,draw=black,fill=black,scale=0.5]{};
\node at (0.25-3.5,-0.75)[circle,draw=black,fill=black,scale=0.5]{};
\node at (0.75-3.5,-0.75)[circle,draw=black,fill=black,scale=0.5]{};
\node at (-0.25-3.5,-0.75)[circle,draw=black,fill=black,scale=0.5]{};
\node at (-0.75-3.5,-0.75)[circle,draw=black,fill=black,scale=0.5]{};
\node at (0.25-3.5,1.25)[circle,draw=black,fill=black,scale=0.5]{};
\node at (-0.25-3.5,1.25)[circle,draw=black,fill=black,scale=0.5]{};
\node at (0.25-3.5,-1.25)[circle,draw=black,fill=black,scale=0.5]{};
\node at (-0.25-3.5,-1.25)[circle,draw=black,fill=black,scale=0.5]{};

\node at (0.45-3.5,0.4){\footnotesize{$23$}};
\node at (0.95-3.5,0.4){\footnotesize{$24$}};
\node at (1.45-3.5,0.4){\footnotesize{$20$}};
\node at (-0.05-3.5,0.4){\footnotesize{$22$}};
\node at (-0.55-3.5,0.4){\footnotesize{$21$}};
\node at (-1.05-3.5,0.4){\footnotesize{$20$}};
\node at (0.45-3.5,-0.1){\footnotesize{$28$}};
\node at (0.95-3.5,-0.1){\footnotesize{$29$}};
\node at (1.45-3.5,-0.1){\footnotesize{$25$}};
\node at (-0.05-3.5,-0.1){\footnotesize{$27$}};
\node at (-0.55-3.5,-0.1){\footnotesize{$26$}};
\node at (-1.05-3.5,-0.1){\footnotesize{$25$}};
\node at (0.45-3.5,0.9){\footnotesize{$18$}};
\node at (0.95-3.5,0.9){\footnotesize{$19$}};
\node at (-0.05-3.5,0.9){\footnotesize{$17$}};
\node at (-0.55-3.5,0.9){\footnotesize{$16$}};
\node at (0.45-3.5,-0.6){\footnotesize{$3$}};
\node at (0.95-3.5,-0.6){\footnotesize{$4$}};
\node at (-0.05-3.5,-0.6){\footnotesize{$2$}};
\node at (-0.55-3.5,-0.6){\footnotesize{$1$}};
\node at (0.45-3.5,1.4){\footnotesize{$13$}};
\node at (-0.05-3.5,1.4){\footnotesize{$12$}};
\node at (0.45-3.5,-1.1){\footnotesize{$8$}};
\node at (-0.05-3.5,-1.1){\footnotesize{$7$}};

\draw (-0.5-7,1) -- (0.5-7,1);
\draw (-1-7,0.5) -- (1-7,0.5);
\draw (-1-7,0) -- (1-7,0);
\draw (-1-7,-0.5) -- (1-7,-0.5);
\draw (-0.5-7,-1) -- (0.5-7,-1);
\draw (1-7,-0.5) -- (1-7,0.5);
\draw (0.5-7,-1) -- (0.5-7,1);
\draw (0-7,-1) -- (0-7,1);
\draw (-0.5-7,-1) -- (-0.5-7,1);
\draw (-1-7,-0.5) -- (-1-7,0.5);

\node at (-0.5-7,1)[circle,draw=black,fill=black,scale=0.5]{};
\node at (0-7,1)[circle,draw=black,fill=black,scale=0.5]{};
\node at (0.5-7,1)[circle,draw=black,fill=black,scale=0.5]{};
\node at (-0.5-7,0.5)[circle,draw=black,fill=black,scale=0.5]{};
\node at (-1-7,0.5)[circle,draw=black,fill=black,scale=0.5]{};
\node at (0-7,0.5)[circle,draw=black,fill=black,scale=0.5]{};
\node at (0.5-7,0.5)[circle,draw=black,fill=black,scale=0.5]{};
\node at (1-7,0.5)[circle,draw=black,fill=black,scale=0.5]{};
\node at (-0.5+7,0)[circle,draw=black,fill=black,scale=0.5]{};
\node at (-1-7,0)[circle,draw=black,fill=black,scale=0.5]{};
\node at (0-7,0)[circle,draw=black,fill=black,scale=0.5]{};
\node at (0.5-7,0)[circle,draw=black,fill=black,scale=0.5]{};
\node at (1-7,0)[circle,draw=black,fill=black,scale=0.5]{};
\node at (-0.5-7,-0.5)[circle,draw=black,fill=black,scale=0.5]{};
\node at (-1-7,-0.5)[circle,draw=black,fill=black,scale=0.5]{};
\node at (0-7,-0.5)[circle,draw=black,fill=black,scale=0.5]{};
\node at (0.5-7,-0.5)[circle,draw=black,fill=black,scale=0.5]{};
\node at (1-7,-0.5)[circle,draw=black,fill=black,scale=0.5]{};
\node at (-0.5-7,-1)[circle,draw=black,fill=black,scale=0.5]{};
\node at (0-7,-1)[circle,draw=black,fill=black,scale=0.5]{};
\node at (0.5-7,-1)[circle,draw=black,fill=black,scale=0.5]{};

\node at (0.2-7,0.15){\footnotesize{$15$}};
\node at (0.7-7,0.15){\footnotesize{$16$}};
\node at (1.2-7,0.15){\footnotesize{$17$}};
\node at (-0.3-7,0.15){\footnotesize{$19$}};
\node at (-0.8-7,0.15){\footnotesize{$18$}};
\node at (0.2-7,0.65){\footnotesize{$10$}};
\node at (0.7-7,0.65){\footnotesize{$11$}};
\node at (1.2-7,0.65){\footnotesize{$12$}};
\node at (-0.3-7,0.65){\footnotesize{$14$}};
\node at (-0.8-7,0.65){\footnotesize{$13$}};
\node at (0.2-7,-0.35){\footnotesize{$20$}};
\node at (0.7-7,-0.35){\footnotesize{$21$}};
\node at (1.2-7,-0.35){\footnotesize{$22$}};
\node at (-0.3-7,-0.35){\footnotesize{$24$}};
\node at (-0.8-7,-0.35){\footnotesize{$23$}};
\node at (0.2-7,1.15){\footnotesize{$5$}};
\node at (0.7-7,1.15){\footnotesize{$6$}};
\node at (-0.3-7,1.15){\footnotesize{$9$}};
\node at (0.2-7,-0.85){\footnotesize{$25$}};
\node at (0.7-7,-0.85){\footnotesize{$26$}};
\node at (-0.3-7,-0.85){\footnotesize{$29$}};

\draw (-0.75-10.5,0.75) -- (0.75-10.5,0.75);
\draw (-0.75-10.5,0.25) -- (0.75-10.5,0.25);
\draw (-0.75-10.5,-0.25) -- (0.75-10.5,-0.25);
\draw (-0.75-10.5,-0.75) -- (0.75-10.5,-0.75);
\draw (-0.75-10.5,-0.75) -- (-0.75-10.5,0.75);
\draw (-0.25-10.5,-0.75) -- (-0.25-10.5,0.75);
\draw (0.25-10.5,-0.75) -- (0.25-10.5,0.75);
\draw (0.75-10.5,-0.75) -- (0.75-10.5,0.75);

\node at (-0.75-10.5,0.75)[circle,draw=black,fill=black,scale=0.5]{};
\node at (-0.75-10.5,0.25)[circle,draw=black,fill=black,scale=0.5]{};
\node at (-0.75-10.5,-0.25)[circle,draw=black,fill=black,scale=0.5]{};
\node at (-0.75-10.5,-0.75)[circle,draw=black,fill=black,scale=0.5]{};
\node at (-0.25-10.5,0.75)[circle,draw=black,fill=black,scale=0.5]{};
\node at (-0.25-10.5,0.25)[circle,draw=black,fill=black,scale=0.5]{};
\node at (-0.25-10.5,-0.25)[circle,draw=black,fill=black,scale=0.5]{};
\node at (-0.25-10.5,-0.75)[circle,draw=black,fill=black,scale=0.5]{};
\node at (0.75-10.5,0.75)[circle,draw=black,fill=black,scale=0.5]{};
\node at (0.75-10.5,0.25)[circle,draw=black,fill=black,scale=0.5]{};
\node at (0.75-10.5,-0.25)[circle,draw=black,fill=black,scale=0.5]{};
\node at (0.75-10.5,-0.75)[circle,draw=black,fill=black,scale=0.5]{};
\node at (0.25-10.5,0.75)[circle,draw=black,fill=black,scale=0.5]{};
\node at (0.25-10.5,0.25)[circle,draw=black,fill=black,scale=0.5]{};
\node at (0.25-10.5,-0.25)[circle,draw=black,fill=black,scale=0.5]{};
\node at (0.25-10.5,-0.75)[circle,draw=black,fill=black,scale=0.5]{};

\node at (0.45-10.5,0.4){\footnotesize{$8$}};
\node at (0.95-10.5,0.4){\footnotesize{$9$}};
\node at (-0.05-10.5,0.4){\footnotesize{$7$}};
\node at (-0.55-10.5,0.4){\footnotesize{$6$}};
\node at (0.45-10.5,0.9){\footnotesize{$3$}};
\node at (0.95-10.5,0.9){\footnotesize{$4$}};
\node at (-0.05-10.5,0.9){\footnotesize{$2$}};
\node at (-0.55-10.5,0.9){\footnotesize{$1$}};
\node at (0.45-10.5,-0.1){\footnotesize{$13$}};
\node at (0.95-10.5,-0.1){\footnotesize{$14$}};
\node at (-0.05-10.5,-0.1){\footnotesize{$12$}};
\node at (-0.55-10.5,-0.1){\footnotesize{$11$}};
\node at (0.45-10.5,-0.6){\footnotesize{$18$}};
\node at (0.95-10.5,-0.6){\footnotesize{$19$}};
\node at (-0.05-10.5,-0.6){\footnotesize{$17$}};
\node at (-0.55-10.5,-0.6){\footnotesize{$16$}};
\end{tikzpicture}
\end{center}
\begin{center}
\begin{tikzpicture}[scale=1.15]
\draw (-1.25+3.5,0.25) -- (1.25+3.5,0.25);
\draw (-0.75+3.5,0.75) -- (0.75+3.5,0.75);
\draw (-0.25+3.5,1.25) -- (0.25+3.5,1.25);
\draw (-1.25+3.5,-0.25) -- (1.25+3.5,-0.25);
\draw (-0.75+3.5,-0.75) -- (0.75+3.5,-0.75);
\draw (-0.25+3.5,-1.25) -- (0.25+3.5,-1.25);

\draw (0.25+3.5,-1.25) -- (0.25+3.5,1.25);
\draw (0.75+3.5,-0.75) -- (0.75+3.5,0.75);
\draw (1.25+3.5,-0.25) -- (1.25+3.5,0.25);
\draw (-0.25+3.5,-1.25) -- (-0.25+3.5,1.25);
\draw (-0.75+3.5,-0.75) -- (-0.75+3.5,0.75);
\draw (-1.25+3.5,-0.25) -- (-1.25+3.5,0.25);

\node at (0.25+3.5,0.25)[circle,draw=black,fill=black,scale=0.5]{};
\node at (0.75+3.5,0.25)[circle,draw=black,fill=black,scale=0.5]{};
\node at (1.25+3.5,0.25)[circle,draw=black,fill=black,scale=0.5]{};
\node at (-0.25+3.5,0.25)[circle,draw=black,fill=black,scale=0.5]{};
\node at (-0.75+3.5,0.25)[circle,draw=black,fill=black,scale=0.5]{};
\node at (-1.25+3.5,0.25)[circle,draw=black,fill=black,scale=0.5]{};
\node at (0.25+3.5,-0.25)[circle,draw=black,fill=black,scale=0.5]{};
\node at (0.75+3.5,-0.25)[circle,draw=black,fill=black,scale=0.5]{};
\node at (1.25+3.5,-0.25)[circle,draw=black,fill=black,scale=0.5]{};
\node at (-0.25+3.5,-0.25)[circle,draw=black,fill=black,scale=0.5]{};
\node at (-0.75+3.5,-0.25)[circle,draw=black,fill=black,scale=0.5]{};
\node at (-1.25+3.5,-0.25)[circle,draw=black,fill=black,scale=0.5]{};
\node at (0.25+3.5,0.75)[circle,draw=black,fill=black,scale=0.5]{};
\node at (0.75+3.5,0.75)[circle,draw=black,fill=black,scale=0.5]{};
\node at (-0.25+3.5,0.75)[circle,draw=black,fill=black,scale=0.5]{};
\node at (-0.75+3.5,0.75)[circle,draw=black,fill=black,scale=0.5]{};
\node at (0.25+3.5,-0.75)[circle,draw=black,fill=black,scale=0.5]{};
\node at (0.75+3.5,-0.75)[circle,draw=black,fill=black,scale=0.5]{};
\node at (-0.25+3.5,-0.75)[circle,draw=black,fill=black,scale=0.5]{};
\node at (-0.75+3.5,-0.75)[circle,draw=black,fill=black,scale=0.5]{};
\node at (0.25+3.5,1.25)[circle,draw=black,fill=black,scale=0.5]{};
\node at (-0.25+3.5,1.25)[circle,draw=black,fill=black,scale=0.5]{};
\node at (0.25+3.5,-1.25)[circle,draw=black,fill=black,scale=0.5]{};
\node at (-0.25+3.5,-1.25)[circle,draw=black,fill=black,scale=0.5]{};

\node at (0.45+3.5,0.4){\footnotesize{$8$}};
\node at (0.95+3.5,0.4){\footnotesize{$9$}};
\node at (1.45+3.5,0.4){\footnotesize{$5$}};
\node at (-0.05+3.5,0.4){\footnotesize{$7$}};
\node at (-0.55+3.5,0.4){\footnotesize{$6$}};
\node at (-1.05+3.5,0.4){\footnotesize{$5$}};
\node at (0.45+3.5,-0.1){\footnotesize{$13$}};
\node at (0.95+3.5,-0.1){\footnotesize{$14$}};
\node at (1.45+3.5,-0.1){\footnotesize{$10$}};
\node at (-0.05+3.5,-0.1){\footnotesize{$12$}};
\node at (-0.55+3.5,-0.1){\footnotesize{$11$}};
\node at (-1.05+3.5,-0.1){\footnotesize{$10$}};
\node at (0.45+3.5,0.9){\footnotesize{$3$}};
\node at (0.95+3.5,0.9){\footnotesize{$4$}};
\node at (-0.05+3.5,0.9){\footnotesize{$2$}};
\node at (-0.55+3.5,0.9){\footnotesize{$1$}};
\node at (0.45+3.5,-0.6){\footnotesize{$18$}};
\node at (0.95+3.5,-0.6){\footnotesize{$19$}};
\node at (-0.05+3.5,-0.6){\footnotesize{$17$}};
\node at (-0.55+3.5,-0.6){\footnotesize{$16$}};
\node at (0.45+3.5,1.4){\footnotesize{$28$}};
\node at (-0.05+3.5,1.4){\footnotesize{$27$}};
\node at (0.45+3.5,-1.1){\footnotesize{$23$}};
\node at (-0.05+3.5,-1.1){\footnotesize{$22$}};

\draw (-0.5+7,1) -- (0.5+7,1);
\draw (-1+7,0.5) -- (1+7,0.5);
\draw (-1+7,0) -- (1+7,0);
\draw (-1+7,-0.5) -- (1+7,-0.5);
\draw (-0.5+7,-1) -- (0.5+7,-1);
\draw (1+7,-0.5) -- (1+7,0.5);
\draw (0.5+7,-1) -- (0.5+7,1);
\draw (0+7,-1) -- (0+7,1);
\draw (-0.5+7,-1) -- (-0.5+7,1);
\draw (-1+7,-0.5) -- (-1+7,0.5);

\node at (-0.5+7,1)[circle,draw=black,fill=black,scale=0.5]{};
\node at (0+7,1)[circle,draw=black,fill=black,scale=0.5]{};
\node at (0.5+7,1)[circle,draw=black,fill=black,scale=0.5]{};
\node at (-0.5+7,0.5)[circle,draw=black,fill=black,scale=0.5]{};
\node at (-1+7,0.5)[circle,draw=black,fill=black,scale=0.5]{};
\node at (0+7,0.5)[circle,draw=black,fill=black,scale=0.5]{};
\node at (0.5+7,0.5)[circle,draw=black,fill=black,scale=0.5]{};
\node at (1+7,0.5)[circle,draw=black,fill=black,scale=0.5]{};
\node at (-0.5+7,0)[circle,draw=black,fill=black,scale=0.5]{};
\node at (-1+7,0)[circle,draw=black,fill=black,scale=0.5]{};
\node at (0+7,0)[circle,draw=black,fill=black,scale=0.5]{};
\node at (0.5+7,0)[circle,draw=black,fill=black,scale=0.5]{};
\node at (1+7,0)[circle,draw=black,fill=black,scale=0.5]{};
\node at (-0.5+7,-0.5)[circle,draw=black,fill=black,scale=0.5]{};
\node at (-1+7,-0.5)[circle,draw=black,fill=black,scale=0.5]{};
\node at (0+7,-0.5)[circle,draw=black,fill=black,scale=0.5]{};
\node at (0.5+7,-0.5)[circle,draw=black,fill=black,scale=0.5]{};
\node at (1+7,-0.5)[circle,draw=black,fill=black,scale=0.5]{};
\node at (-0.5+7,-1)[circle,draw=black,fill=black,scale=0.5]{};
\node at (0+7,-1)[circle,draw=black,fill=black,scale=0.5]{};
\node at (0.5+7,-1)[circle,draw=black,fill=black,scale=0.5]{};

\node at (0.2+7,0.15){\footnotesize{$15$}};
\node at (0.7+7,0.15){\footnotesize{$16$}};
\node at (1.2+7,0.15){\footnotesize{$17$}};
\node at (-0.3+7,0.15){\footnotesize{$19$}};
\node at (-0.8+7,0.15){\footnotesize{$18$}};
\node at (0.2+7,0.65){\footnotesize{$10$}};
\node at (0.7+7,0.65){\footnotesize{$11$}};
\node at (1.2+7,0.65){\footnotesize{$12$}};
\node at (-0.3+7,0.65){\footnotesize{$14$}};
\node at (-0.8+7,0.65){\footnotesize{$13$}};
\node at (0.2+7,-0.35){\footnotesize{$20$}};
\node at (0.7+7,-0.35){\footnotesize{$21$}};
\node at (1.2+7,-0.35){\footnotesize{$22$}};
\node at (-0.3+7,-0.35){\footnotesize{$24$}};
\node at (-0.8+7,-0.35){\footnotesize{$23$}};
\node at (0.2+7,1.15){\footnotesize{$5$}};
\node at (0.7+7,1.15){\footnotesize{$6$}};
\node at (-0.3+7,1.15){\footnotesize{$9$}};
\node at (0.2+7,-0.85){\footnotesize{$25$}};
\node at (0.7+7,-0.85){\footnotesize{$26$}};
\node at (-0.3+7,-0.85){\footnotesize{$29$}};

\draw (-0.75+10.5,0.75) -- (0.75+10.5,0.75);
\draw (-0.75+10.5,0.25) -- (0.75+10.5,0.25);
\draw (-0.75+10.5,-0.25) -- (0.75+10.5,-0.25);
\draw (-0.75+10.5,-0.75) -- (0.75+10.5,-0.75);
\draw (-0.75+10.5,-0.75) -- (-0.75+10.5,0.75);
\draw (-0.25+10.5,-0.75) -- (-0.25+10.5,0.75);
\draw (0.25+10.5,-0.75) -- (0.25+10.5,0.75);
\draw (0.75+10.5,-0.75) -- (0.75+10.5,0.75);

\node at (-0.75+10.5,0.75)[circle,draw=black,fill=black,scale=0.5]{};
\node at (-0.75+10.5,0.25)[circle,draw=black,fill=black,scale=0.5]{};
\node at (-0.75+10.5,-0.25)[circle,draw=black,fill=black,scale=0.5]{};
\node at (-0.75+10.5,-0.75)[circle,draw=black,fill=black,scale=0.5]{};
\node at (-0.25+10.5,0.75)[circle,draw=black,fill=black,scale=0.5]{};
\node at (-0.25+10.5,0.25)[circle,draw=black,fill=black,scale=0.5]{};
\node at (-0.25+10.5,-0.25)[circle,draw=black,fill=black,scale=0.5]{};
\node at (-0.25+10.5,-0.75)[circle,draw=black,fill=black,scale=0.5]{};
\node at (0.75+10.5,0.75)[circle,draw=black,fill=black,scale=0.5]{};
\node at (0.75+10.5,0.25)[circle,draw=black,fill=black,scale=0.5]{};
\node at (0.75+10.5,-0.25)[circle,draw=black,fill=black,scale=0.5]{};
\node at (0.75+10.5,-0.75)[circle,draw=black,fill=black,scale=0.5]{};
\node at (0.25+10.5,0.75)[circle,draw=black,fill=black,scale=0.5]{};
\node at (0.25+10.5,0.25)[circle,draw=black,fill=black,scale=0.5]{};
\node at (0.25+10.5,-0.25)[circle,draw=black,fill=black,scale=0.5]{};
\node at (0.25+10.5,-0.75)[circle,draw=black,fill=black,scale=0.5]{};

\node at (0.45+10.5,0.4){\footnotesize{$23$}};
\node at (0.95+10.5,0.4){\footnotesize{$24$}};
\node at (-0.05+10.5,0.4){\footnotesize{$22$}};
\node at (-0.55+10.5,0.4){\footnotesize{$21$}};
\node at (0.45+10.5,0.9){\footnotesize{$18$}};
\node at (0.95+10.5,0.9){\footnotesize{$19$}};
\node at (-0.05+10.5,0.9){\footnotesize{$17$}};
\node at (-0.55+10.5,0.9){\footnotesize{$16$}};
\node at (0.45+10.5,-0.1){\footnotesize{$28$}};
\node at (0.95+10.5,-0.1){\footnotesize{$29$}};
\node at (-0.05+10.5,-0.1){\footnotesize{$27$}};
\node at (-0.55+10.5,-0.1){\footnotesize{$26$}};
\node at (0.45+10.5,-0.6){\footnotesize{$3$}};
\node at (0.95+10.5,-0.6){\footnotesize{$4$}};
\node at (-0.05+10.5,-0.6){\footnotesize{$2$}};
\node at (-0.55+10.5,-0.6){\footnotesize{$1$}};

\node at (3.5,-1.75){$\mathcal{F}_{k+1}$};
\node at (7,-1.75){$\mathcal{F}_{k+2}$};
\node at (10.5,-1.75){$\mathcal{F}_{k+3}$};
\end{tikzpicture}
\end{center}
\caption{The colors of the vertices at distance at most 3 of a fixed vertex $u$ of color $0$ in the layer $k$ of the graph $\mathcal{F}$ (square: $u$; circle: vertex at distance $d$ of $u$, $1\le d\le 3$).}
\label{colorpower3}
\end{figure}
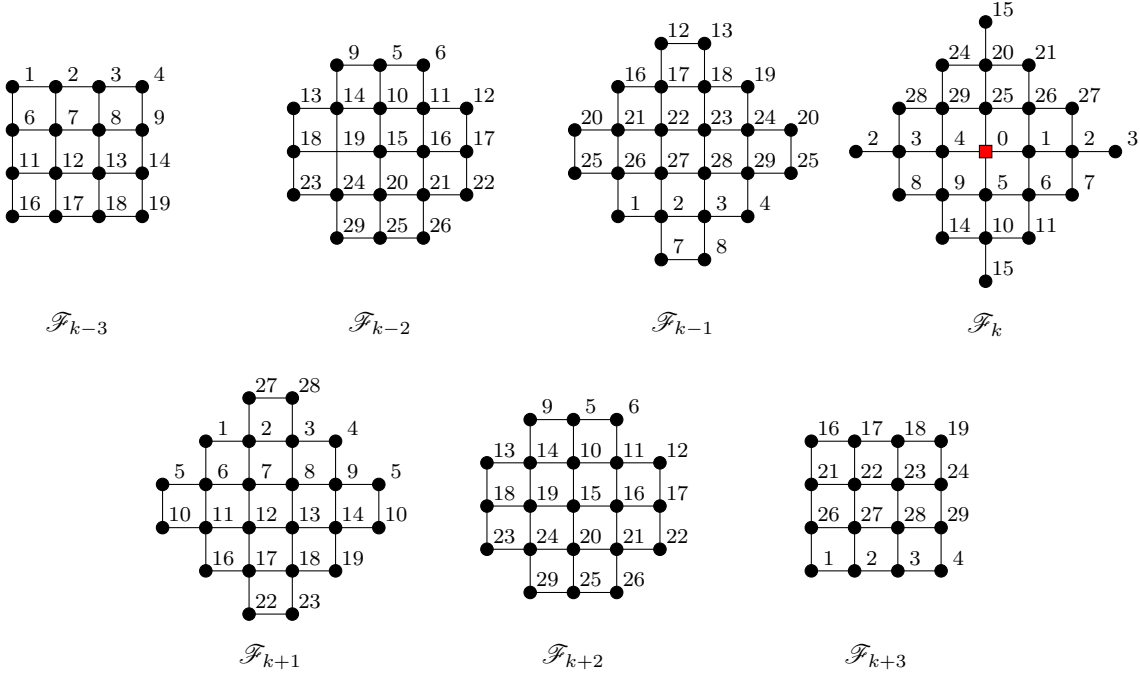
\end{proof}

For vertices $u_1$, $u_2$, $u_3$ and $u_4$ of $\mathcal{F}$, we denote by $T(u_1,u_2,u_3,u_4)$, the vertex set $\{ u\in\mathcal{F}|\ uv\in E(\mathcal{F}),\ v\in\{u_1,u_2,u_3,u_4\} \}$. Note that in the case $\{u_1,u_2,u_3,u_4\}$ induces a complete graph in $\mathcal{F}$, $T(u_1,u_2,u_3,u_4)$ induces in $\mathcal{F}$ a graph isomorphic to the graph $B_1$ from Lemma \ref{le2}. Consequently, in this case, $\text{diam}(T(u_1,u_2,u_3,u_4))=3$.

\begin{lemma}\label{balancetonlemme}
If there exists a 28-coloring of $\mathcal{F}^3$, then the following is true about every vertex $(i,j,k)$:
\begin{enumerate}
\item[(i)] one vertex among $(i+3, j, k+2)$, $(i+2.5,j+0.5,k+3)$ and $(i+2.5,j-0.5,k+3)$ has the same color than $(i,j,k)$;
\item[(ii)] one vertex among $(i-3, j, k+2)$, $(i-2.5,j+0.5,k+3)$ and $(i-2.5,j-0.5,k+3)$ has the same color than $(i,j,k)$;
\item[(iii)] one vertex among $(i, j-3, k+2)$, $(i+0.5,j-2.5,k+3)$ and $(i-0.5,j-2.5,k+3)$ has the same color than $(i,j,k)$;
\item[(iv)] one vertex among $(i, j-3, k-2)$, $(i+0.5,j-2.5,k-3)$ and $(i-0.5,j-2.5,k-3)$ has the same color than $(i,j,k)$.

\end{enumerate}
\end{lemma}

\begin{figure}[t]
\begin{center}
\begin{tikzpicture}[scale=0.9]

\draw (0-1,2.5/3-0.5) -- (2,2.5/3-0.5);
\draw (0-0.5,5/3-0.5) -- (3-0.5,5/3-0.5);
\draw (0.5-0.5,2) -- (3,2);
\draw (-1,2.5/3-0.5) -- (0,2);
\draw (0,2.5/3-0.5) -- (1,2);
\draw (1,2.5/3-0.5) -- (2,2);
\draw (2,2.5/3-0.5) -- (3,2);

\draw (0+3.5,0) -- (2+3.5,0);
\draw (0.5+3.5,2.5/3) -- (2.5+3.5,2.5/3);
\draw (1+3.5,5/3) -- (3+3.5,5/3);
\draw (1.5+3.5,2.5) -- (3.5+3.5,2.5);
\draw (0+3.5,0) -- (1.5+3.5,2.5);
\draw (1+3.5,0) -- (2.5+3.5,2.5);
\draw (2+3.5,0) -- (3.5+3.5,2.5);

\draw (0-1+8.5,2.5/3-0.5) -- (2+8.5,2.5/3-0.5);
\draw (0-0.5+8.5,5/3-0.5) -- (3-0.5+8.5,5/3-0.5);
\draw (0.5-0.5+8.5,2) -- (3+8.5,2);
\draw (-1+8.5,2.5/3-0.5) -- (0+8.5,2);
\draw (0+8.5,2.5/3-0.5) -- (1+8.5,2);
\draw (1+8.5,2.5/3-0.5) -- (2+8.5,2);
\draw (2+8.5,2.5/3-0.5) -- (3+8.5,2);

\draw (0+12,0) -- (2+12,0);
\draw (0.5+12,2.5/3) -- (2.5+12,2.5/3);
\draw (1+12,5/3) -- (3+12,5/3);
\draw (1.5+12,2.5) -- (3.5+12,2.5);
\draw (0+12,0) -- (1.5+12,2.5);
\draw (1+12,0) -- (2.5+12,2.5);
\draw (2+12,0) -- (3.5+12,2.5);

\draw[ultra thick, color=red] (0,2.5/3-0.5) -- (1,2);
\draw[ultra thick, color=red] (1,2.5/3-0.5) -- (2,2);
\draw[ultra thick, color=red] (1-1,2.5/3-0.5) -- (1,2.5/3-0.5);
\draw[ultra thick, color=red] (1-0.5,5/3-0.5) -- (2-0.5,5/3-0.5);
\draw[ultra thick, color=red] (1.5-0.5,2) -- (2,2);

\draw[ultra thick, color=red]  (0.5+3.5,2.5/3) -- (1+3.5,5/3);
\draw[ultra thick, color=red]  (1+3.5,0) -- (2.5+3.5,2.5);
\draw[ultra thick, color=red]  (2.5+3.5,2.5/3) -- (3+3.5,5/3);
\draw[ultra thick, color=red]  (1.5/3+3.5,2.5/3) -- (2+1.5/3+3.5,2.5/3);
\draw[ultra thick, color=red]  (1+3.5,5/3) -- (3+3.5,5/3);

\draw[ultra thick, color=red] (0+8.5,2.5/3-0.5) -- (1+8.5,2);
\draw[ultra thick, color=red] (1+8.5,2.5/3-0.5) -- (2+8.5,2);
\draw[ultra thick, color=red] (1-1+8.5,2.5/3-0.5) -- (1+8.5,2.5/3-0.5);
\draw[ultra thick, color=red] (-0.5+8.5,5/3-0.5) -- (3-0.5+8.5,5/3-0.5);
\draw[ultra thick, color=red] (1.5-0.5+8.5,2) -- (2+8.5,2);

\draw[ultra thick, color=red]  (0.5+12,2.5/3) -- (1+12,5/3);
\draw[ultra thick, color=red]  (1.5+12,2.5/3) -- (2+12,5/3);
\draw[ultra thick, color=red]  (2.5+12,2.5/3) -- (3+12,5/3);
\draw[ultra thick, color=red]  (1.5/3+12,2.5/3) -- (2+1.5/3+12,2.5/3);
\draw[ultra thick, color=red]  (1+12,5/3) -- (3+12,5/3);

\node at (0,2.5/3-0.5)[circle,draw=black,fill=black,scale=0.5]{};
\node at (1,2.5/3-0.5)[circle,draw=black,fill=black,scale=0.5]{};
\node at (-0.5,5/3-0.5)[regular polygon, regular polygon sides=4,draw=black,fill=red,scale=0.5]{};
\node at (0.5,5/3-0.5)[circle,draw=black,fill=black,scale=0.5]{};
\node at (1.5,5/3-0.5)[circle,draw=black,fill=black,scale=0.5]{};
\node at (1,2)[circle,draw=black,fill=black,scale=0.5]{};
\node at (2,2)[circle,draw=black,fill=black,scale=0.5]{};

\node at (2.5+3.5,2.5)[circle,draw=black,fill=black,scale=0.5]{};
\node at (3+3.5,5/3)[circle,draw=black,fill=black,scale=0.5]{};
\node at (2.5+3.5,2.5/3)[circle,draw=black,fill=black,scale=0.5]{};
\node at (3+3.5-1,5/3)[circle,draw=black,fill=black,scale=0.5]{};
\node at (2.5+3.5-1,2.5/3)[circle,draw=black,fill=black,scale=0.5]{};
\node at (3+3.5-2,5/3)[circle,draw=black,fill=black,scale=0.5]{};
\node at (4.5+3.5-4,2.5/3)[circle,draw=black,fill=black,scale=0.5]{};
\node at (1+3.5,0)[circle,draw=black,fill=black,scale=0.5]{};
\node at (0+8.5,2.5/3-0.5)[circle,draw=black,fill=black,scale=0.5]{};
\node at (1+8.5,2.5/3-0.5)[circle,draw=black,fill=black,scale=0.5]{};
\node at (-0.5+8.5,5/3-0.5)[circle,draw=black,fill=black,scale=0.5]{};
\node at (0.5+8.5,5/3-0.5)[circle,draw=black,fill=black,scale=0.5]{};
\node at (1.5+8.5,5/3-0.5)[circle,draw=black,fill=black,scale=0.5]{};
\node at (2.5+8.5,5/3-0.5)[regular polygon, regular polygon sides=3,draw=black,fill=green,scale=0.35]{};
\node at (1+8.5,2)[circle,draw=black,fill=black,scale=0.5]{};
\node at (2+8.5,2)[circle,draw=black,fill=black,scale=0.5]{};

\node at (3+12,5/3)[regular polygon, regular polygon sides=3,draw=black,fill=green,scale=0.35]{};
\node at (1.5+12,2.5/3)[circle,draw=black,fill=black,scale=0.5]{};
\node at (3+12-1,5/3)[circle,draw=black,fill=black,scale=0.5]{};
\node at (3.5+12-1,2.5/3)[regular polygon, regular polygon sides=3,draw=black,fill=green,scale=0.35]{};
\node at (3+12-2,5/3)[circle,draw=black,fill=black,scale=0.5]{};
\node at (4.5+12-4,2.5/3)[circle,draw=black,fill=black,scale=0.5]{};

\node at (0.5,-0.5){$\mathcal{F}_{i}$};
\node at (4.5,-0.5){$\mathcal{F}_{i+1}$};
\node at (9,-0.5){$\mathcal{F}_{i+2}$};
\node at (4.5+8.5,-0.5){$\mathcal{F}_{i+3}$};

\end{tikzpicture}
\end{center}
\caption{The vertex $(i,j,k)$ (the square), the vertices in $B$ at distance at most $3$ from $(i,j,k)$ (circles) and the vertices at distance at least $4$ from $(i,j,k)$ in $B$ (triangles; thick lines: edges between vertices in $B$).}
\label{petitfrereveutgrandirtropvite}
\end{figure}
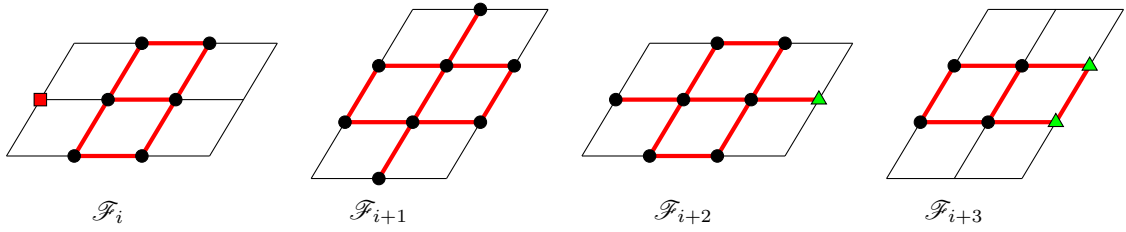

\begin{proof}
(i) Let $B=T((i+1.5,j-0.5,k+1),$ $(i+1.5,j+0.5,k+1),$ $(i+1,j,k+2),$ $(i+2,j,k+2))$.
Let $B'=B\cup\{(i,j,k)\}$.  Since $B$ induces in $\mathcal{F}$ a graph isomorphic to the graph $B_1$ from Lemma \ref{le2}, we have $|B|=28$. Note that $(i,j,k)$ is at distance at most $3$ from every vertex of $B$ except $(i+3, j, k+2)$, $(i+2.5,j+0.5,k+3)$ and $(i+2.5,j-0.5,k+3)$. 
Figure \ref{petitfrereveutgrandirtropvite} illustrates the set $B$, the vertex $(i,j,k)$ and the vertices of $B$ at distance at most $3$ or not from $(i,j,k)$.
Since $|B'|=29$, in every possible 28-coloring of $\mathcal{F}^3$, the vertex $(i,j,k)$ has the same color than a vertex among $(i+3, j, k+2)$, $(i+2.5,j+0.5,k+3)$ and $(i+2.5,j-0.5,k+3)$.\newline
(ii), (iii) and (iv) By symmetry, the proof is totally analogous to the proof of Case (i) (the set $B$ should be chosen differently).\newline
\end{proof}
Note that the lower bound of Proposition \ref{oddbound} (the value 28) corresponds to the size of the largest clique in $\mathcal{F}^3$. However, the next result shows that the value of $\chi(\mathcal{F}^{3})$ is larger than 28.
\begin{prop}
$ \chi(\mathcal{F}^{3})\ge 29$.
\end{prop}

\begin{proof}
Let $B=T((0,0,0),$ $(0,1,0),$ $(-0.5,0.5,1),$ $(0.5,0.5,1))$ and let $B'=T((2,0,0),$ $(2,1,0),$ $(1.5,0.5,1),$ $(2.5,0.5,1))$. Note that $|B|=28$ and $|B'|=28$ and as in Lemma~\ref{balancetonlemme} both $B$ and $B'$ induce in $\mathcal{F}$ a graph isomorphic to the graph $B_1$ from Lemma \ref{le2}. 
Suppose that $\mathcal{F}^3$ is 28-colorable. By hypothesis, both $B$ and $B'$ should contain a vertex of every color.
Let $u=(1,2,0)$, $v_1=(-1,0,0)$ and $v_2=(3,0,0)$.
Moreover, let $X_1= \{(-0.5,-0.5,-1),$ $(2.5,-0.5,-1),$ $(0,-1,0)$, $(2,-1,0) \}$, $X_2=(-0.5,-0.5,1))$, $(0.5,-0.5,1\}$, $X_3=\{(-1.5,0.5,1)),$ $(-1,0,2)\}$ and $X_4=\{(3.5,0.5,1),(3,0,2)\}$.
Figure \ref{ecoutecabb} illustrates the sets $B$, $B'$, $X_1$, $X_2$, $X_3$ and $X_4$ and the vertices $u$, $v_1$ and $v_2$.

Remark that $u$ is at distance at most $3$ from every vertex of $B\cup B'\setminus (X_1\cup X_2\cup X_3\cup X_4\cup \{v_1,v_2\})$. 
Figure \ref{ecoutecabb} illustrates the vertices of $B\cup B'$ at distance at most $3$ of $u$.

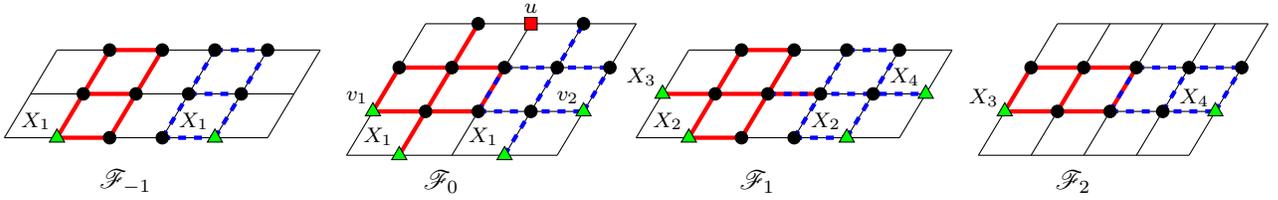
\begin{figure}[t]
\begin{center}
\begin{tikzpicture}[scale=0.7]

\draw (0-1,2.5/3-0.5) -- (4,2.5/3-0.5);
\draw (0-0.5,5/3-0.5) -- (5-0.5,5/3-0.5);
\draw (0.5-0.5,2) -- (5,2);
\draw (-1,2.5/3-0.5) -- (0,2);
\draw (0,2.5/3-0.5) -- (1,2);
\draw (1,2.5/3-0.5) -- (2,2);
\draw (2,2.5/3-0.5) -- (3,2);
\draw (3,2.5/3-0.5) -- (4,2);
\draw (4,2.5/3-0.5) -- (5,2);

\draw (5.5,0) -- (4+5.5,0);
\draw (1.5/3+5.5,2.5/3) -- (4+1.5/3+5.5,2.5/3);
\draw (1+5.5,5/3) -- (5+5.5,5/3);
\draw (1.5+5.5,2.5) -- (1.5+4+5.5,2.5);
\draw (0+5.5,0) -- (1.5+5.5,2.5);
\draw (1+5.5,0) -- (2.5+5.5,2.5);
\draw (2+5.5,0) -- (3.5+5.5,2.5);
\draw (3+5.5,0) -- (4.5+5.5,2.5);
\draw (4+5.5,0) -- (5.5+5.5,2.5);

\draw (0+1.5/3+10.5,2.5/3-0.5) -- (5+1.5/3+10.5,2.5/3-0.5);
\draw (0+1+10.5,5/3-0.5) -- (6+10.5,5/3-0.5);
\draw (1.5+10.5,2.5-0.5) -- (6.5+10.5,2.5-0.5);
\draw (0.5+10.5,2.5/3-0.5) -- (1.5+10.5,2);
\draw (1.5+10.5,2.5/3-0.5) -- (2.5+10.5,2);
\draw (2.5+10.5,2.5/3-0.5) -- (3.5+10.5,2);
\draw (3.5+10.5,2.5/3-0.5) -- (4.5+10.5,2);
\draw (4.5+10.5,2.5/3-0.5) -- (5.5+10.5,2);
\draw (5.5+10.5,2.5/3-0.5) -- (6.5+10.5,2);

\draw (0+17.5,0) -- (4+17.5,0);
\draw (0+1.5/3+17.5,2.5/3) -- (4+1.5/3+17.5,2.5/3);
\draw (0+1+17.5,5/3) -- (5+17.5,5/3);
\draw (1.5+17.5,2.5) -- (5.5+17.5,2.5);
\draw (0+17.5,0) -- (1.5+17.5,2.5);
\draw (1+17.5,0) -- (2.5+17.5,2.5);
\draw (2+17.5,0) -- (3.5+17.5,2.5);
\draw (3+17.5,0) -- (4.5+17.5,2.5);
\draw (4+17.5,0) -- (5.5+17.5,2.5);

\draw[ultra thick, color=red] (0,2.5/3-0.5) -- (1,2);
\draw[ultra thick, color=red] (1,2.5/3-0.5) -- (2,2);
\draw[ultra thick, color=blue, dashed] (2,2.5/3-0.5) -- (3,2);
\draw[ultra thick, color=blue, dashed] (3,2.5/3-0.5) -- (4,2);
\draw[ultra thick, color=red] (1-1,2.5/3-0.5) -- (1,2.5/3-0.5);
\draw[ultra thick, color=red] (1-0.5,5/3-0.5) -- (2-0.5,5/3-0.5);
\draw[ultra thick, color=red] (1.5-0.5,2) -- (2,2);
\draw[ultra thick, color=blue, dashed] (3-1,2.5/3-0.5) -- (3,2.5/3-0.5);
\draw[ultra thick, color=blue, dashed] (3-0.5,5/3-0.5) -- (4-0.5,5/3-0.5);
\draw[ultra thick, color=blue, dashed] (3.5-0.5,2) -- (4,2);

\draw[ultra thick, color=red]  (0.5+5.5,2.5/3) -- (1+5.5,5/3);
\draw[ultra thick, color=red]  (1+5.5,0) -- (2.5+5.5,2.5);
\draw[ultra thick, color=red]  (2.5+5.5,2.5/3) -- (3+5.5,5/3);
\draw[ultra thick, color=blue, dashed]  (2.5+5.5,2.5/3) -- (3+5.5,5/3);
\draw[ultra thick, color=blue, dashed]  (3+5.5,0) -- (4.5+5.5,2.5);
\draw[ultra thick, color=blue, dashed]  (4.5+5.5,2.5/3) -- (5+5.5,5/3);
\draw[ultra thick, color=red]  (1.5/3+5.5,2.5/3) -- (2+1.5/3+5.5,2.5/3);
\draw[ultra thick, color=red]  (1+5.5,5/3) -- (3+5.5,5/3);
\draw[ultra thick, color=blue, dashed]  (2+1.5/3+5.5,2.5/3) -- (4+1.5/3+5.5,2.5/3);
\draw[ultra thick, color=blue, dashed] (3+5.5,5/3) -- (5+5.5,5/3);

\draw[ultra thick, color=red] (0+12,2.5/3-0.5) -- (1+12,2);
\draw[ultra thick, color=red] (1+12,2.5/3-0.5) -- (2+12,2);
\draw[ultra thick, color=blue, dashed] (2+12,2.5/3-0.5) -- (3+12,2);
\draw[ultra thick, color=blue, dashed] (3+12,2.5/3-0.5) -- (4+12,2);
\draw[ultra thick, color=red] (1-1+12,2.5/3-0.5) -- (1+12,2.5/3-0.5);
\draw[ultra thick, color=red] (0-0.5+12,5/3-0.5) -- (3-0.5+12,5/3-0.5);
\draw[ultra thick, color=red] (1.5-0.5+12,2) -- (2+12,2);
\draw[ultra thick, color=blue, dashed] (3-1+12,2.5/3-0.5) -- (3+12,2.5/3-0.5);
\draw[ultra thick, color=blue, dashed] (2-0.5+12,5/3-0.5) -- (5-0.5+12,5/3-0.5);
\draw[ultra thick, color=blue, dashed] (3.5-0.5+12,2) -- (4+12,2);

\draw[ultra thick, color=red]  (0.5+17.5,2.5/3) -- (1+17.5,5/3);
\draw[ultra thick, color=red]  (1.5+17.5,2.5/3) -- (2+17.5,5/3);
\draw[ultra thick, color=red]  (2.5+17.5,2.5/3) -- (3+17.5,5/3);
\draw[ultra thick, color=blue, dashed]  (2.5+17.5,2.5/3) -- (3+17.5,5/3);
\draw[ultra thick, color=blue, dashed]  (3.5+17.5,2.5/3) -- (4+17.5,5/3);
\draw[ultra thick, color=blue, dashed]  (4.5+17.5,2.5/3) -- (5+17.5,5/3);
\draw[ultra thick, color=red]  (1.5/3+17.5,2.5/3) -- (2+1.5/3+17.5,2.5/3);
\draw[ultra thick, color=red]  (1+17.5,5/3) -- (3+17.5,5/3);
\draw[ultra thick, color=blue, dashed]  (2+1.5/3+17.5,2.5/3) -- (4+1.5/3+17.5,2.5/3);
\draw[ultra thick, color=blue, dashed] (3+17.5,5/3) -- (5+17.5,5/3);

\node at (0,2.5/3-0.5)[regular polygon, regular polygon sides=3,draw=black,fill=green,scale=0.35]{};
\node at (1,2.5/3-0.5)[circle,draw=black,fill=black,scale=0.5]{};
\node at (2,2.5/3-0.5)[circle,draw=black,fill=black,scale=0.5]{};
\node at (3,2.5/3-0.5)[regular polygon, regular polygon sides=3,draw=black,fill=green,scale=0.35]{};
\node at (0.5,5/3-0.5)[circle,draw=black,fill=black,scale=0.5]{};
\node at (1.5,5/3-0.5)[circle,draw=black,fill=black,scale=0.5]{};
\node at (2.5,5/3-0.5)[circle,draw=black,fill=black,scale=0.5]{};
\node at (3.5,5/3-0.5)[circle,draw=black,fill=black,scale=0.5]{};
\node at (1,2)[circle,draw=black,fill=black,scale=0.5]{};
\node at (2,2)[circle,draw=black,fill=black,scale=0.5]{};
\node at (3,2)[circle,draw=black,fill=black,scale=0.5]{};
\node at (4,2)[circle,draw=black,fill=black,scale=0.5]{};

\node at (3.5+5.5,2.5)[regular polygon, regular polygon sides=4,draw=black,fill=red,scale=0.5]{};
\node at (2.5+5.5,2.5)[circle,draw=black,fill=black,scale=0.5]{};
\node at (4.5+5.5,2.5)[circle,draw=black,fill=black,scale=0.5]{};
\node at (3+5.5,5/3)[circle,draw=black,fill=black,scale=0.5]{};
\node at (2.5+5.5,2.5/3)[circle,draw=black,fill=black,scale=0.5]{};
\node at (3+5.5-1,5/3)[circle,draw=black,fill=black,scale=0.5]{};
\node at (2.5+5.5-1,2.5/3)[circle,draw=black,fill=black,scale=0.5]{};
\node at (3+5.5+1,5/3)[circle,draw=black,fill=black,scale=0.5]{};
\node at (2.5+5.5+1,2.5/3)[circle,draw=black,fill=black,scale=0.5]{};
\node at (3+5.5-2,5/3)[circle,draw=black,fill=black,scale=0.5]{};
\node at (3+5.5+2,5/3)[circle,draw=black,fill=black,scale=0.5]{};
\node at (4.5+5.5-4,2.5/3)[regular polygon, regular polygon sides=3,draw=black,fill=green,scale=0.35]{};
\node at (4.5+5.5,2.5/3)[regular polygon, regular polygon sides=3,draw=black,fill=green,scale=0.35]{};
\node at (1+5.5,0)[regular polygon, regular polygon sides=3,draw=black,fill=green,scale=0.35]{};
\node at (3+5.5,0)[regular polygon, regular polygon sides=3,draw=black,fill=green,scale=0.35]{};

\node at (1.5+11-0.5,2.5/3-0.5)[regular polygon, regular polygon sides=3,draw=black,fill=green,scale=0.35]{};
\node at (2.5+11-0.5,2.5/3-0.5)[circle,draw=black,fill=black,scale=0.5]{};
\node at (3.5+11-0.5,2.5/3-0.5)[circle,draw=black,fill=black,scale=0.5]{};
\node at (4.5+11-0.5,2.5/3-0.5)[regular polygon, regular polygon sides=3,draw=black,fill=green,scale=0.35]{};

\node at (1+11-0.5,5/3-0.5)[regular polygon, regular polygon sides=3,draw=black,fill=green,scale=0.35]{};
\node at (2+11-0.5,5/3-0.5)[circle,draw=black,fill=black,scale=0.5]{};
\node at (3+11-0.5,5/3-0.5)[circle,draw=black,fill=black,scale=0.5]{};
\node at (4+11-0.5,5/3-0.5)[circle,draw=black,fill=black,scale=0.5]{};
\node at (5+11-0.5,5/3-0.5)[circle,draw=black,fill=black,scale=0.5]{};
\node at (6+11-0.5,5/3-0.5)[regular polygon, regular polygon sides=3,draw=black,fill=green,scale=0.35]{};
\node at (2.5+11-0.5,2)[circle,draw=black,fill=black,scale=0.5]{};
\node at (3.5+11-0.5,2)[circle,draw=black,fill=black,scale=0.5]{};
\node at (4.5+11-0.5,2)[circle,draw=black,fill=black,scale=0.5]{};
\node at (5.5+11-0.5,2)[circle,draw=black,fill=black,scale=0.5]{};

\node at (3+17.5,5/3)[circle,draw=black,fill=black,scale=0.5]{};
\node at (2.5+17.5,2.5/3)[circle,draw=black,fill=black,scale=0.5]{};
\node at (3+17.5-1,5/3)[circle,draw=black,fill=black,scale=0.5]{};
\node at (2.5+17.5-1,2.5/3)[circle,draw=black,fill=black,scale=0.5]{};
\node at (3+17.5+1,5/3)[circle,draw=black,fill=black,scale=0.5]{};
\node at (2.5+17.5+1,2.5/3)[circle,draw=black,fill=black,scale=0.5]{};
\node at (3+17.5-2,5/3)[circle,draw=black,fill=black,scale=0.5]{};
\node at (3+17.5+2,5/3)[circle,draw=black,fill=black,scale=0.5]{};
\node at (4.5+17.5-4,2.5/3)[regular polygon, regular polygon sides=3,draw=black,fill=green,scale=0.35]{};
\node at (4.5+17.5,2.5/3)[regular polygon, regular polygon sides=3,draw=black,fill=green,scale=0.35]{};

\node at (3.5+5.5,2.8){\footnotesize{$u$}};
\node at (5.7,1.1){\footnotesize{$v_{1}$}};
\node at (4+5.7,1.1){\footnotesize{$v_{2}$}};
\node at (0.3-0.7,0.65){\footnotesize{$X_{1}$}};
\node at (3.3-0.7,0.65){\footnotesize{$X_{1}$}};
\node at (1+5.1,0.35){\footnotesize{$X_{1}$}};
\node at (3+5.1,0.35){\footnotesize{$X_{1}$}};
\node at (1+10.1,0.65+2.5/3){\footnotesize{$X_{3}$}};
\node at (6+10.1,0.65+2.5/3){\footnotesize{$X_{4}$}};
\node at (1.5+10.1,0.65){\footnotesize{$X_{2}$}};
\node at (4.5+10.1,0.65){\footnotesize{$X_{2}$}};
\node at (5.6+12,1.1){\footnotesize{$X_{3}$}};
\node at (4+5.6+12,1.1){\footnotesize{$X_{4}$}};
\node at (1.3,-0.5){$\mathcal{F}_{-1}$};
\node at (7.3,-0.5){$\mathcal{F}_{0}$};
\node at (13.3,-0.5){$\mathcal{F}_{1}$};
\node at (19.3,-0.5){$\mathcal{F}_{2}$};

\end{tikzpicture}
\end{center}
\caption{The vertex $u$ (the square), the vertices in $B\cup B'$ at distance at most $3$ from $u$ (circles) and the vertices at distance at least $4$ from $u$ in $B\cup B'$ (triangles; thick lines: edges between vertices in $B$; dashed lines: edges between vertices in $B'$).}
\label{ecoutecabb}
\end{figure}

Suppose, by contradiction, that there exists a $28$-coloring $c$ of $\mathcal{F}^{3}$. We begin by showing that $u$ can not have the same color than the one of any vertex from $X_1$, $X_2$, $X_3$ and $X_4$. 
Afterward, we show that we have a contradiction in the case $u$, $v_1$ and $v_2$ have the same color. 
Consequently, $c(u)\neq c(v_1)$ or $c(u)\neq c(v_2)$, and $u$ can not have the same color than the one of any vertex from either $B$ or $B'$ (which both contain a vertex of every color).

First, by Lemma \ref{balancetonlemme}.(iv), one vertex from $\{(1,-1,-2),$ $(1.5,-0.5,-3),$ $(0.5,-0.5,-3)\}$ has the same color than $u$. Note that every vertex of $\{(1,-1,-2),$ $(1.5,-0.5,-3),$ $(0.5,-0.5,-3)\}$ is at distance at most $3$ from $(-0.5,-0.5,-1)$, $(2.5,-0.5,-1)$, $(0,-1,0)$ and $(2,-1,0)$. Consequently, $c(u)\neq c(v)$, if $v\in X_1$.

Second, by Lemma \ref{balancetonlemme}.(iii), one vertex from $\{(1,-1,2),$ $(1.5,-0.5,3),$ $(0.5,-0.5,3)\}$ has the same color than $u$. Note that every vertex of $\{(1,-1,2),$ $(1.5,-0.5,3),$ $(0.5,-0.5,3)\}$ is at distance at most $3$ from $(-0.5,-0.5,1)$ and $(2.5,-0.5,1)$. Consequently, $c(u)\neq c(v)$, if $v\in X_2$.

Third, by Lemma \ref{balancetonlemme}.(ii), one vertex from $\{(-2,2,2),$ $(-1.5,2.5,3),$ $(-1.5,1.5,3)\}$ has the same color than $u$. Note that every vertex of $\{(-2,2,2),$ $(-1.5,2.5,3),$ $(-1.5,1.5,3)\}$ is at distance at most $3$ from $(-1.5,0.5,1)$ and $(-1,0,2)$. Consequently, $c(u)\neq c(v)$, if $v\in X_3$.

Fourth, by Lemma \ref{balancetonlemme}.(i), one vertex from $\{(4,2,2),$ $(3.5,2.5,3),$ $(3.5,1.5,3)\}$ has the same color than $u$. Note that every vertex of $\{(4,2,2),$ $(3.5,2.5,3),$ $(3.5,1.5,3)\}$ is at distance at most $3$ from $(3.5,0.5,1)$ and $(3,0,2)$. Consequently, $c(u)\neq c(v)$, if $v\in X_4$.

Now, by the previous arguments, we know that the color of $u$ is different from the colors of the vertices of $B\setminus\{v_1\}$ and of $B'\setminus\{v_2\}$.
Thus, since  both $B$ and $B'$ should contain a vertex of every color, we have $c(u)=c(v_1)$ and $c(u)=c(v_2)$.
By Lemma \ref{balancetonlemme}.(i), one vertex from $\{(2,0,2),$ $(1.5,0.5,3),$ $(1.5,-0.5,3)\}$ has the same color than $v_1$. Note that $d_{\mathcal{F}}(v_2,$ $(2,0,2))=2$, $d_{\mathcal{F}}(v_2,$ $(1.5,-0.5,3))=3$ and $d_{\mathcal{F}}(v_2,$ $(1.5,0.5,3))=3$. Therefore, we have a contradiction with the fact that $c(v_1)=c(v_2)$.

To conclude, $u$ can not have the same color that any vertex of either $B$ or $B'$ and since $|B|=28$ and $|B'|=28$, $\mathcal{F}^3$ is not 28-colorable.

\end{proof}
\section{Chromatic number of the $d^{\text{th}}$ power of $\mathcal{F}_{0,1}$ and $\mathcal{F}_{0,2}$}

In this section, we give bounds on the chromatic number of two subgraphs of $\mathcal{F}$. These subgraphs are obtained by restricting the graph to the vertices of two or three consecutive layers of $\mathcal{F}$. Note that we achieved to obtain the exact value of the chromatic number of the $d^{\text{th}}$ power of these subgraphs in more cases that for $\mathcal{F}$.

\subsection{Chromatic number of the $d^{\text{th}}$ power of $\mathcal{F}_{0,1}$}

We begin this subsection by calculating the size of a largest complete subgraph in $\mathcal{F}_{0,1}^d$. From this value, we derive the following lower bound on the chromatic number of $\mathcal{F}_{0,1}^{d}$.

\begin{prop}
For every positive integer $d$, we have $\chi(\mathcal{F}_{0,1}^{d})\ge (d+1)^2$.
\end{prop}
\begin{proof}
First, suppose $d$ is even. Let $A_{\ell}=\{(i,j,k)\in V(\mathcal{F}_{0,1})|\ d_{\mathcal{F}}((0,0,0)(i,j,k))= \ell\}$.
We claim that $|A_\ell|=8 \ell$, for $\ell \ge 1$. We proceed by induction to prove it. Note that $|A_1|=8$. Suppose, by induction, that $|A_\ell|=8 \ell$.
It can be easily remarked that both $|V(\mathcal{F}_0)\cap A_{\ell+1}|=|V(\mathcal{F}_0)\cap A_{\ell}|+4$ and $|V(\mathcal{F}_1)\cap A_{\ell+1}|=|V(\mathcal{F}_1)\cap A_{\ell}|+4$. Thus $|A_{\ell+1}|=8\ell+8$.
By calculation $|\cup^\ell_{i=0} A_{i}|=\sum_{i=1}^\ell (8i) +1=4 \ell (\ell+1)+1$. Consequently, since  $\text{diam}(\mathcal{F}_{0,1}[\cup^\ell_{i=0} A_{i}])\le 2\ell$, we have: $$\chi(\mathcal{F}_{0,1}^{d})\ge 4 \frac{d}{2} \left(\frac{d}{2}+1\right)+1=d(d+2)+1=(d+1)^2.$$

Second, suppose $d$ is odd. Let $D_{0}=\{(0,0,0),$ $(1,0,0),$ $(0.5,0.5,1),$ $(0.5,-0.5,1)\}$ and let $D_{\ell }=\{ u\in V(\mathcal{F}_{0,1})|\ \min_{v\in D_{0}}(d_{\mathcal{F}}(u,v))=\ell \}$, for $\ell \ge 0$. We claim that $|D_{\ell}|=8 \ell+4$, for $\ell\ge 0$. We proceed by induction to prove it. Note that $|D_0|=4$ and $|D_1|=12$.  Suppose, by induction, that $|D_\ell|=8 \ell+4$.
It can be easily remarked that both $|V(\mathcal{F}_0)\cap D_{\ell+1}|=|V(\mathcal{F}_0)\cap D_{\ell}|+4$ and $|V(\mathcal{F}_1)\cap D_{\ell+1}|=|V(\mathcal{F}_1)\cap D_{\ell}|+4$. Thus $|D_{\ell+1}|=8\ell+8$.
By calculation, $|\cup^\ell_{i=0} D_{i}|=\sum_{i=0}^\ell (8i +4)=4 (\ell+1)^2$. Consequently, since  $\text{diam}(\mathcal{F}_{0,1}[\cup^\ell_{i=0} D_{i}])\le 2\ell+1$, we have : $$\chi(\mathcal{F}_{0,1}^{d})\ge 4\left(\frac{(d-1)}{2}+1\right)^2=(d+1)^2.$$
\end{proof}
In the following proposition, we show that, in contrast with $\mathcal{F}^d$, the size of the largest complete subgraph in $\mathcal{F}_{0,1}^{d}$ corresponds to the chromatic number of $\mathcal{F}_{0,1}^{d}$, for every integer $d$.

\begin{prop}
For every positive integer $d$, we have $\chi(\mathcal{F}_{0,1}^{d})\le (d+1)^2$.
\end{prop}
\begin{proof}
It can be remarked that the coloring from Theorem \ref{upbound} restricted to $\mathcal{F}_{0,1}^d$ is a coloring using $(d+1)^2$ colors, when $d$ is odd. Then, we can suppose that $d$ is even and, consequently, that $d/2$ is an integer.
We set the following function: $\forall (i,j,k)\in V(\mathcal{F}),$
$$c(i,j,k)) = i+(2d+1) j \pmod{(d+1)^2}. $$

We claim that $c$ is a coloring function of $\mathcal{F}_{0,1}^{d}$.

It remains to prove that every two vertices $(i,j,k)$ and $(i',j',k')$ satisfying $c((i,j,k))=c((i',j',k'))$ are such that $d_{\mathcal{F}}((i,j,k),(i',j',k'))>d$. Without loss of generality, we suppose that $i'=0$, $j'=0$ and $k'=0$.

Note that the vertices $(i,j,k)$ at distance at most $d$ from $(0,0,0)$ are such that $1\le |i|+|j|\le d$. Also, note that $c((0,0,0))=c((i,j,k))$ implies $i+(2d+1) j\equiv 0\pmod{(d+1)^2}$. 
Thus, it remains to show that for every possible value of $i$ and $j$, $1\le |i|+|j|\le d$, we have $i+(2d+1) j\not \equiv 0\pmod{(d+1)^2}$.

If $|j|\le  d/2$, then $ -d^2-d/2\le (2d+1)|j|\le d^2+d/2$. 
Thus, since $|i|+|j|\le d$, we have $-(d+1)^2 < -d^2-d \le i+(2d+1)j\le d^2+d < (d+1)^2$ and since $|i|+|j|\neq 0$, we have $i+(2d+1) j\not \equiv 0\pmod{(d+1)^2}$.

Moreover, if $d/2+1 \le |j|\le d$, then $d^2+5d/2+1 \le (2d+1)|j|\le 2d^2+d$ and consequently, since $|i|+|j|\le d$, we have  $d^2+2d+2 \le i+(2d+1)|j|\le 2d^2+3d/2-1$. 
Since $(d+1)^2 < d^2+2d+2$ and $2d^2+3d/2-1<2 (d+1)^2$, we have $i+(2d+1) j\not \equiv 0\pmod{(d+1)^2}$.

\end{proof}

\subsection{Chromatic number of the $d^{\text{th}}$ power of $\mathcal{F}_{0,2}$}
We begin this subsection by giving a general upper bound (using Theorem \ref{upbound}) on the chromatic number of $\mathcal{F}_{0,2}^{d}$.
\begin{prop}
For every positive integer $d$, we have $\chi(\mathcal{F}_{0,2}^{d})\le 3\lceil (d+1)/2 \rceil ^2$.
\end{prop}
\begin{proof}
It can be remarked that the coloring from Theorem \ref{upbound} restricted to $\mathcal{F}_{0,2}^d$ is a coloring of $\mathcal{F}$ using $3\lceil (d+1)^2/2 \rceil $ colors.
\end{proof}

We continue this subsection by determining the size of a largest complete subgraph in $\mathcal{F}_{0,2}^d$. From this value, we derive the following lower bound on the chromatic number of $\mathcal{F}_{0,2}^{d}$.
\begin{prop}\label{whitefire}
For every positive integer $d$, we have $\chi(\mathcal{F}_{0,2}^{d})\ge 3(\lfloor (d+1)^2/2 \rfloor )-2$.
\end{prop}
\begin{proof}
Let $A'_{\ell }=\{(i,j,k)\in V(\mathcal{F}_{0,2})|\ d_{\mathcal{F}_{0,2}}((0,0,1),(i,j,k)= \ell \}$.
We prove by induction that $|A'_\ell|=12 \ell$, for $\ell\ge 1$.
Note that we have $|A'_0|=1$ and $|A'_1|=12$. Suppose that $|A'_\ell|=12\ell$.
We can easily notice that, for each of the three layers of $\mathcal{F}_{0,2}$, there are four more vertices in $A'_{\ell+1}$, than in $A'_\ell$. 
Consequently, $|A'_{\ell+1}|=12\ell+12=12 (\ell+1)$. 
Let $A_{\ell}=\cup_{0\le i\le \ell} A'_i$. By calculation,  $|A_\ell|=\sum_{i=0}^{\ell} |A'_i|=1+12 \sum_{i=1}^{\ell} i=1+6\ell (\ell+1)=6\ell^2+6\ell+1$. 
Since $\text{diam}(\mathcal{F}_{0,2}[A_\ell])=2\ell$, we have $\chi(\mathcal{F}_{0,2}^d)\ge 3/2 d^2+3d+1=3\lceil (d+1)^2/2 \rceil -1$, for an even $d$.

Let $B'_{\ell}=\{(i,j,k)\in V(\mathcal{F}_{0,2})|\ d_{\mathcal{F}}(u,(i,j,k)= \ell,\ u\in\{(0,0,0),$ $(1,0,0),$ $(0.5,-0.5,1),$ $(0.5,0.5,1)\} \}$. We prove by induction that $B'_{\ell}=12\ell+6$, for $\ell \ge 1$.
Note that we have $|B'_0|=4$ and $|B'_1|=18$. Suppose that $|B'_\ell|=12\ell+6$.
We can easily notice that, for each of the three layers of $\mathcal{F}_{0,2}$, there are four more vertices in $|B'_{\ell+1}|$, than in $|B'_\ell|$. 
Consequently, $|B'_{\ell+1}|=12\ell+12+6=12 (\ell+1)+6$.
Let $B_{\ell}=\cup_{0\le i\le \ell} B'_i$. By calculation,  $|B_\ell|=\sum_{i=0}^{\ell} |B'_i|=4+12 \sum_{i=1}^{\ell} i+6\ell=4+6\ell (\ell+1)+6\ell=6\ell^2+12\ell+4$.
Since $\text{diam}(\mathcal{F}_{0,2}[B_\ell])=2\ell+1$, we have $\chi(\mathcal{F}_{0,2}^d)\ge 6/4 (d-1)^2+6 (d-1)+4=3/2 d^2+3d-1/2=3 \lceil (d+1)^2/2 \rceil -2$, for an odd $d$.
\end{proof}
Now, we  prove that the lower bound from Proposition \ref{whitefire} can be improved.
In order to do this, we begin by proving three Lemmas that will be used in the proof of Theorem \ref{sharknado}.

\begin{lemma}\label{le2661}
If there exists a $23$-coloring $c$ of $\mathcal{F}_{0,2}^3$, then for every integers $i$ and $j$, $c((i,j,0))=c((i+3,j,2))$ or $c((i+3,j,0))=c((i,j,2))$.
\end{lemma}
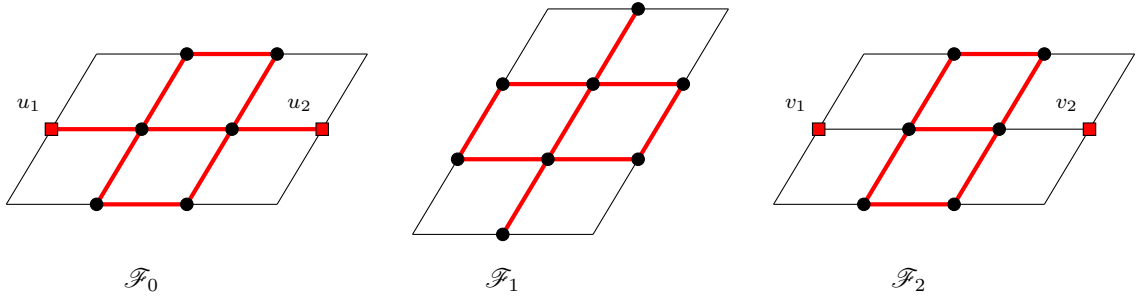
\begin{figure}[t]
\begin{center}
\begin{tikzpicture}[scale=1.2]

\draw (0-1,2.5/3-0.5) -- (2,2.5/3-0.5);
\draw (0-0.5,5/3-0.5) -- (3-0.5,5/3-0.5);
\draw (0.5-0.5,2) -- (3,2);
\draw (-1,2.5/3-0.5) -- (0,2);
\draw (0,2.5/3-0.5) -- (1,2);
\draw (1,2.5/3-0.5) -- (2,2);
\draw (2,2.5/3-0.5) -- (3,2);

\draw (0+3.5,0) -- (2+3.5,0);
\draw (0.5+3.5,2.5/3) -- (2.5+3.5,2.5/3);
\draw (1+3.5,5/3) -- (3+3.5,5/3);
\draw (1.5+3.5,2.5) -- (3.5+3.5,2.5);
\draw (0+3.5,0) -- (1.5+3.5,2.5);
\draw (1+3.5,0) -- (2.5+3.5,2.5);
\draw (2+3.5,0) -- (3.5+3.5,2.5);

\draw (0-1+8.5,2.5/3-0.5) -- (2+8.5,2.5/3-0.5);
\draw (0-0.5+8.5,5/3-0.5) -- (3-0.5+8.5,5/3-0.5);
\draw (0.5-0.5+8.5,2) -- (3+8.5,2);
\draw (-1+8.5,2.5/3-0.5) -- (0+8.5,2);
\draw (0+8.5,2.5/3-0.5) -- (1+8.5,2);
\draw (1+8.5,2.5/3-0.5) -- (2+8.5,2);
\draw (2+8.5,2.5/3-0.5) -- (3+8.5,2);

\draw[ultra thick, color=red] (0,2.5/3-0.5) -- (1,2);
\draw[ultra thick, color=red] (1,2.5/3-0.5) -- (2,2);
\draw[ultra thick, color=red] (1-1,2.5/3-0.5) -- (1,2.5/3-0.5);
\draw[ultra thick, color=red] (-0.5,5/3-0.5) -- (3-0.5,5/3-0.5);
\draw[ultra thick, color=red] (1.5-0.5,2) -- (2,2);

\draw[ultra thick, color=red]  (0.5+3.5,2.5/3) -- (1+3.5,5/3);
\draw[ultra thick, color=red]  (1+3.5,0) -- (2.5+3.5,2.5);
\draw[ultra thick, color=red]  (2.5+3.5,2.5/3) -- (3+3.5,5/3);
\draw[ultra thick, color=red]  (1.5/3+3.5,2.5/3) -- (2+1.5/3+3.5,2.5/3);
\draw[ultra thick, color=red]  (1+3.5,5/3) -- (3+3.5,5/3);

\draw[ultra thick, color=red] (0+8.5,2.5/3-0.5) -- (1+8.5,2);
\draw[ultra thick, color=red] (1+8.5,2.5/3-0.5) -- (2+8.5,2);
\draw[ultra thick, color=red] (1-1+8.5,2.5/3-0.5) -- (1+8.5,2.5/3-0.5);
\draw[ultra thick, color=red] (1-0.5+8.5,5/3-0.5) -- (2-0.5+8.5,5/3-0.5);
\draw[ultra thick, color=red] (1.5-0.5+8.5,2) -- (2+8.5,2);

\node at (0,2.5/3-0.5)[circle,draw=black,fill=black,scale=0.5]{};
\node at (1,2.5/3-0.5)[circle,draw=black,fill=black,scale=0.5]{};
\node at (-0.5,5/3-0.5)[regular polygon, regular polygon sides=4,draw=black,fill=red,scale=0.5]{};
\node at (0.5,5/3-0.5)[circle,draw=black,fill=black,scale=0.5]{};
\node at (1.5,5/3-0.5)[circle,draw=black,fill=black,scale=0.5]{};
\node at (2.5,5/3-0.5) [regular polygon, regular polygon sides=4,draw=black,fill=red,scale=0.5]{};
\node at (1,2)[circle,draw=black,fill=black,scale=0.5]{};
\node at (2,2)[circle,draw=black,fill=black,scale=0.5]{};

\node at (2.5+3.5,2.5)[circle,draw=black,fill=black,scale=0.5]{};
\node at (3+3.5,5/3)[circle,draw=black,fill=black,scale=0.5]{};
\node at (2.5+3.5,2.5/3)[circle,draw=black,fill=black,scale=0.5]{};
\node at (3+3.5-1,5/3)[circle,draw=black,fill=black,scale=0.5]{};
\node at (2.5+3.5-1,2.5/3)[circle,draw=black,fill=black,scale=0.5]{};
\node at (3+3.5-2,5/3)[circle,draw=black,fill=black,scale=0.5]{};
\node at (4.5+3.5-4,2.5/3)[circle,draw=black,fill=black,scale=0.5]{};
\node at (1+3.5,0)[circle,draw=black,fill=black,scale=0.5]{};
\node at (0+8.5,2.5/3-0.5)[circle,draw=black,fill=black,scale=0.5]{};
\node at (1+8.5,2.5/3-0.5)[circle,draw=black,fill=black,scale=0.5]{};
\node at (-0.5+8.5,5/3-0.5)[regular polygon, regular polygon sides=4,draw=black,fill=red,scale=0.5]{};
\node at (0.5+8.5,5/3-0.5)[circle,draw=black,fill=black,scale=0.5]{};
\node at (1.5+8.5,5/3-0.5)[circle,draw=black,fill=black,scale=0.5]{};
\node at (2.5+8.5,5/3-0.5)[regular polygon, regular polygon sides=4,draw=black,fill=red,scale=0.5]{};
\node at (1+8.5,2)[circle,draw=black,fill=black,scale=0.5]{};
\node at (2+8.5,2)[circle,draw=black,fill=black,scale=0.5]{};

\node at (0.95-1.7,0.6+2.5/3){\footnotesize{$u_{1}$}};
\node at (3.95-1.7,0.6+2.5/3){\footnotesize{$u_{2}$}};

\node at (0.95+8.5-1.7,0.6+2.5/3){\footnotesize{$v_{1}$}};
\node at (3.95+8.5-1.7,0.6+2.5/3){\footnotesize{$v_{2}$}};

\node at (0.5,-0.5){$\mathcal{F}_{0}$};
\node at (4.5,-0.5){$\mathcal{F}_{1}$};
\node at (9,-0.5){$\mathcal{F}_{2}$};

\end{tikzpicture}
\end{center}
\caption{The vertices in $B$ (circle) and the vertices $u_1$, $u_2$, $v_1$ and $v_2$ (squares; thick lines: edges between vertices in $B$).}
\label{blackninja}
\end{figure}

\begin{proof}
Let $c$ be a $23$-coloring of $\mathcal{F}_{0,2}$ and let $u_1=(i,j,0)$, $u_2=(i,j,2)$, $v_1=(i+3,j,0)$ and $v_2=(i+3,j,2)$. Suppose that $c(u_1)\neq c(v_2)$.
Let $B=T((i+1,j,0),$ $(i+2,j,0),$ $(i+1.5,j-0.5,1),$ $(i+1.5,j+0.5,1))$. 
Figure \ref{blackninja} illustrates the set $B$ and the vertices $u_1$, $u_2$, $v_1$ and $v_2$. 
Also, note that $B$ induces a subgraph of diameter $3$ in $\mathcal{F}_{0,2}$ and that $|B|=22$. For this reason, $22$ colors are required to color $B$. Since $v_2$ is at distance at most $3$ of every vertex of $B\setminus \{u_1 \}$,  23 colors are required to color $B\cup  \{v_2 \}$.
Also, note that $u_2$ is at distance at most $3$ of every vertex of $B\setminus \{v_1 \}$ and that $d_{\mathcal{F}_{0,2}}(u_2,v_2)=3$. Since $u_2$ is at distance at most $3$ of every vertex of $B\setminus \{v_1 \}\cup \{v_2 \}$, the only possibility is that $c(u_2)=c(v_1)$ which conclude our proof.
The proof is totally similar in the case $c(v_1)\neq c(u_2)$.
\end{proof}

The following Lemma is an intermediate step in order to prove Lemma \ref{le2667} that will be used several times in the proof of Theorem~\ref{sharknado}.
\begin{lemma}\label{le2665}
If there exists a $23$-coloring $c$ of $\mathcal{F}_{0,2}^3$ and integers $i$ and $j$ such that $c((i,j,0))=c((i+3,j,2))$ and $c((i+3,j,0))=c((i,j,2))$, then $c((i,j+1,0))=c((i+3,j+1,2))$ and $c((i+3,j+1,0))=c((i,j+1,2))$.
\end{lemma}

\begin{proof}
Let $u_1=(i,j,0)$, $u_2=(i,j,2)$, $v_1=(i+3,j,0)$, $v_2=(i+3,j,2)$, $u'_1=(i,j+1,0)$, $u'_2=(i,j+1,2)$, $v'_1=(i+3,j+1,0)$ and $v'_2=(i+3,j+1,2)$.
Let $c$ be a 23-coloring of $\mathcal{F}_{0,2}^3$ such that $c(u_1)=c(v_2)$ and $c(u_2)=c(v_1)$.
Suppose by contradiction that $c(u'_1)\neq c(v'_2)$ or $c(v'_1)\neq c(u'_2)$.
Let $B=T((i+1,j+1,0),$ $(i+2,j+1,0),$ $(i+1.5,j+0.5,1),$ $(i+1.5,j+1.5,1))$.
Note that $B$ induces in $\mathcal{F}_{0,2}$ a subgraph of diameter $3$ containing $22$ vertices. 
The following is true:
\begin{itemize}
\item[i)] $v'_2$ is at distance at most $3$ from any vertex of $B\setminus\{ u'_1\}$; 
\item[ii)] $u'_2$ is at distance at most $3$ from any vertex of $B\setminus\{ v'_1\}$;
\item[iii)] every vertex of $(B\setminus\{(i+1.5,j+2.5,1)\})\cup\{u'_2,v'_2\}$ is at distance at most $3$ from either $u_1$ or $v_2$ and at distance at most $3$ from either $v_1$ or $u_2$.
\end{itemize}
Figure \ref{transmorpher} illustrates why iii) holds.
Note that $B\cup\{u'_2,v'_2\}$ contains, by i) and ii), vertices of $23$ different colors.
Let $a=c(u_1)=c(v_2)$ and $b=c(v_1)=c(u_2)$.
\begin{figure}[t]
\begin{center}
\begin{tikzpicture}[scale=1.2]

\draw (0-1,2.5/3-0.5) -- (2,2.5/3-0.5);
\draw (0-0.5,5/3-0.5) -- (3-0.5,5/3-0.5);
\draw (0.5-0.5,2) -- (3,2);
\draw (-1,2.5/3-0.5) -- (0,2);
\draw (0,2.5/3-0.5) -- (1,2);
\draw (1,2.5/3-0.5) -- (2,2);
\draw (2,2.5/3-0.5) -- (3,2);

\draw (0+3.5,0) -- (2+3.5,0);
\draw (0.5+3.5,2.5/3) -- (2.5+3.5,2.5/3);
\draw (1+3.5,5/3) -- (3+3.5,5/3);
\draw (1.5+3.5,2.5) -- (3.5+3.5,2.5);
\draw (0+3.5,0) -- (1.5+3.5,2.5);
\draw (1+3.5,0) -- (2.5+3.5,2.5);
\draw (2+3.5,0) -- (3.5+3.5,2.5);

\draw (0-1+8.5,2.5/3-0.5) -- (2+8.5,2.5/3-0.5);
\draw (0-0.5+8.5,5/3-0.5) -- (3-0.5+8.5,5/3-0.5);
\draw (0.5-0.5+8.5,2) -- (3+8.5,2);
\draw (-1+8.5,2.5/3-0.5) -- (0+8.5,2);
\draw (0+8.5,2.5/3-0.5) -- (1+8.5,2);
\draw (1+8.5,2.5/3-0.5) -- (2+8.5,2);
\draw (2+8.5,2.5/3-0.5) -- (3+8.5,2);

\draw[ultra thick, color=red] (0,2.5/3-0.5) -- (1,2);
\draw[ultra thick, color=red] (1,2.5/3-0.5) -- (2,2);
\draw[ultra thick, color=red] (1-1,2.5/3-0.5) -- (1,2.5/3-0.5);
\draw[ultra thick, color=red] (-0.5,5/3-0.5) -- (3-0.5,5/3-0.5);
\draw[ultra thick, color=red] (1.5-0.5,2) -- (2,2);

\draw[ultra thick, color=red]  (0.5+3.5,2.5/3) -- (1+3.5,5/3);
\draw[ultra thick, color=red]  (1+3.5,0) -- (2.5+3.5,2.5);
\draw[ultra thick, color=red]  (2.5+3.5,2.5/3) -- (3+3.5,5/3);
\draw[ultra thick, color=red]  (1.5/3+3.5,2.5/3) -- (2+1.5/3+3.5,2.5/3);
\draw[ultra thick, color=red]  (1+3.5,5/3) -- (3+3.5,5/3);

\draw[ultra thick, color=red] (0+8.5,2.5/3-0.5) -- (1+8.5,2);
\draw[ultra thick, color=red] (1+8.5,2.5/3-0.5) -- (2+8.5,2);
\draw[ultra thick, color=red] (1-1+8.5,2.5/3-0.5) -- (1+8.5,2.5/3-0.5);
\draw[ultra thick, color=red] (-0.5+8.5,5/3-0.5) -- (3-0.5+8.5,5/3-0.5);
\draw[ultra thick, color=red] (1.5-0.5+8.5,2) -- (2+8.5,2);

\node at (0-1,2.5/3-0.5)[regular polygon, regular polygon sides=4,draw=black,fill=red,scale=0.5]{};
\node at (3-1,2.5/3-0.5)[regular polygon, regular polygon sides=4,draw=black,fill=red,scale=0.5]{};
\node at (0,2.5/3-0.5)[circle,draw=black,fill=black,scale=0.5]{};
\node at (1,2.5/3-0.5)[circle,draw=black,fill=black,scale=0.5]{};
\node at (-0.5,5/3-0.5)[circle,draw=black,fill=black,scale=0.5]{};
\node at (0.5,5/3-0.5)[circle,draw=black,fill=black,scale=0.5]{};
\node at (1.5,5/3-0.5)[circle,draw=black,fill=black,scale=0.5]{};
\node at (2.5,5/3-0.5) [regular polygon, regular polygon sides=3,draw=black,fill=green,scale=0.35]{};
\node at (1,2)[circle,draw=black,fill=black,scale=0.5]{};
\node at (2,2)[regular polygon, regular polygon sides=3,draw=black,fill=green,scale=0.35]{};

\node at (3+3.5,5/3)[regular polygon, regular polygon sides=3,draw=black,fill=green,scale=0.35]{};
\node at (2.5+3.5,2.5/3)[regular polygon, regular polygon sides=3,draw=black,fill=green,scale=0.35]{};
\node at (3+3.5-1,5/3)[circle,draw=black,fill=black,scale=0.5]{};
\node at (2.5+3.5-1,2.5/3)[circle,draw=black,fill=black,scale=0.5]{};
\node at (3+3.5-2,5/3)[circle,draw=black,fill=black,scale=0.5]{};
\node at (4.5+3.5-4,2.5/3)[circle,draw=black,fill=black,scale=0.5]{};
\node at (1+3.5,0)[circle,draw=black,fill=black,scale=0.5]{};

\node at (8.5-1,2.5/3-0.5)[regular polygon, regular polygon sides=4,draw=black,fill=red,scale=0.5]{};
\node at (3+7.5,2.5/3-0.5)[regular polygon, regular polygon sides=4,draw=black,fill=red,scale=0.5]{};
\node at (0+8.5,2.5/3-0.5)[regular polygon, regular polygon sides=3,draw=black,fill=green,scale=0.35]{};
\node at (1+8.5,2.5/3-0.5)[regular polygon, regular polygon sides=3,draw=black,fill=green,scale=0.35]{};
\node at (-0.5+8.5,5/3-0.5)[circle,draw=black,fill=black,scale=0.5]{};
\node at (0.5+8.5,5/3-0.5)[regular polygon, regular polygon sides=3,draw=black,fill=green,scale=0.35]{};
\node at (1.5+8.5,5/3-0.5)[regular polygon, regular polygon sides=3,draw=black,fill=green,scale=0.35]{};
\node at (2.5+8.5,5/3-0.5)[regular polygon, regular polygon sides=3,draw=black,fill=green,scale=0.35]{};
\node at (1+8.5,2)[circle,draw=black,fill=black,scale=0.5]{};
\node at (2+8.5,2)[regular polygon, regular polygon sides=3,draw=black,fill=green,scale=0.35]{};

\node at (0.95-1.7,0.6+2.5/3){\footnotesize{$u'_{1}$}};
\node at (3.95-1.7,0.6+2.5/3){\footnotesize{$u'_{2}$}};

\node at (0.95+8.5-1.7,0.6+2.5/3){\footnotesize{$v'_{1}$}};
\node at (3.95+8.5-1.7,0.6+2.5/3){\footnotesize{$v'_{2}$}};
\node at (6,2.7){\tiny{$(1.5,2.5,1)$}};

\node at (0.95-2.2,0.6){\footnotesize{$u_{1}$}};
\node at (3.95-2.2,0.6){\footnotesize{$u_{2}$}};

\node at (0.95+8-1.7,0.6){\footnotesize{$v_{1}$}};
\node at (3.95+8-1.7,0.6){\footnotesize{$v_{2}$}};

\node at (0.5,-0.5){$\mathcal{F}_{0}$};
\node at (4.5,-0.5){$\mathcal{F}_{1}$};
\node at (9,-0.5){$\mathcal{F}_{2}$};

\end{tikzpicture}
\end{center}
\caption{The vertices in $B\cup\{u'_2,v'_2\}$ at distance at most $3$ from $u_1$ (circles), the vertices in $B\cup\{u'_2,v'_2\}$ at distance at most $3$ from $v_2$ (triangles) and the vertices $u_1$, $u_2$, $v_1$ and $v_2$ (squares; thick lines: edges between vertices in $B\cup\{u'_2,v'_2\}$).}
\label{transmorpher}
\end{figure}
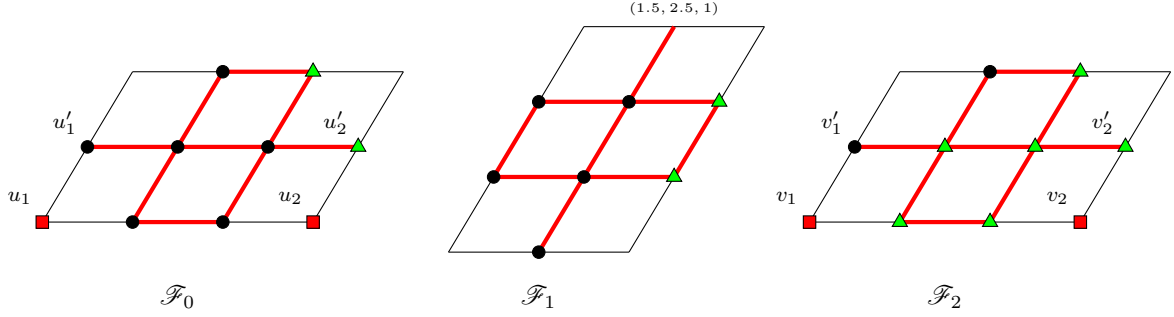

By iii), both the colors $a$ and $b$ are not used for vertices of $(B\setminus\{(i+1.5,j+2.5,1)\})\cup\{u'_2,v'_2\}$. Moreover, since $c((i+1.5,j+2.5,1))\neq a$ or $c((i+1.5,j+2.5,1))\neq b$ and by iii), the color $a$ or $b$ is different than each of the 23 other colors and we have a contradiction with the fact that $\mathcal{F}_{0,2}^3$ is $23$-colorable.
\end{proof}
Note that the following Lemma is really useful since it implies that if $B=T((i,j,0),$ $(i+1,j,0),$ $(i+0.5,j-0.5,1),$ $(i+0.5,j+0.5,1))$ and there exists a $23$-coloring $c$ of $\mathcal{F}_{0,2}^3$, then $B\cup \{(i-1,j,2), (i+2,j,2)\}$ contains $23$ different colors. This is due to the fact that $(i-1,j,2)$ is at distance at most $3$ of every vertex of $B\setminus \{(i+2,j,0)\}$ and that $(i+2,j,2)$ is at distance at most $3$ of every vertex of $B\setminus \{(i-1,j,0)\}$.
\begin{lemma}\label{le2667}
There exists no $23$-coloring $c$ of $\mathcal{F}_{0,2}^3$ for which there exists integers $i$ and $j$ such that $c((i,j,0))=c((i+3,j,2))$ and $c((i+3,j,0))=c((i,j,2))$.
\end{lemma}
\begin{proof}
Let $u_1=(i,j,0)$, $u_2=(i,j,2)$, $v_1=(i+3,j,0)$ and $v_2=(i+3,j,2)$.
Suppose, by contradiction, that there is a $23$-coloring $c$ of $\mathcal{F}_{0,2}^3$ such that $c(u_1)= c(v_2)$ and $c(v_1)= c(u_2)$.
By Lemma \ref{le2665}, we have $c((i,j+1,0))= c((i+3,j+1,2))$ and $c((i+3,j+1,0))= c((i,j+1,2))$ and also $c((i,j+2,0))= c((i+3,j+2,2))$ and $c((i+3,j+2,0))= c((i,j+2,2))$. 
Let $B=T((i+1,j+1,0),$ $(i+2,j+1,0),$ $(i+1.5,j+0.5,1),$ $(i+1.5,j+1.5,1))$.

Note that $B$ induces in $\mathcal{F}_{0,2}$ a subgraph of diameter $3$ containing $22$ vertices. 
The following is true :
\begin{itemize}
\item[i)] every vertex of $B\setminus\{(i+1.5,j-2.5,1)\}$ is at distance at most $3$ from either $u_1$ or $v_2$ and at most $3$ from either $v_1$ or $u_2$;
\item[ii)] every vertex of $B\setminus\{(i+1.5,j+0.5,1)\}$ is at distance at most $3$ from either $(i,j+2,0)$ or $(i+3,j+2,2)$ and at most $3$ from either $(i+3,j+2,0)$ or $(i,j+2,2)$.
\end{itemize}

Let $a=c(u_1)$, $b=c(u_2)$, $a'=c((i,j+2,0))$ and $b'=((i,j+2,2))$.
Since $u_1$, $u_2$, $(i,j+2,0)$ and $(i,j,2)$ are pairwise at distance at most $3$, the colors $a$, $b$, $a'$, $b'$ are all different.
By i), either $a$ or $b$ is a color not used for vertices in $B$. Also, by ii), either $a'$ or $b'$ is a color not used for vertices in $B$. Therefore, since $B$ contains vertices of $22$ colors, we have to use two other colors. Consequently, we have a contradiction with the fact that $\mathcal{F}_{0,2}^3$ is $23$-colorable.

\end{proof}
We finish this paper by determining the exact value of the chromatic number of $\mathcal{F}_{0,2}^3$.
Note that the exact value of the chromatic number of $\mathcal{F}^3$ remains undetermined. 
\begin{theorem}\label{sharknado}
We have $\chi(\mathcal{F}_{0,2}^3)\ge 24$.
\end{theorem}
\begin{proof}

Suppose by contradiction that there exists a 23-coloring $c$ of $\mathcal{F}_{0,2}^3$.
Let $u_1=(0,0,0)$, $u_2=(0,0,2)$, $v_1=(3,0,0)$, $v_2=(3,0,2)$, $x_1=(0,3,0)$, $x_2=(0,3,2)$, $y_1=(-3,0,0)$ and $y_2=(-3,0,2)$. 
Let  $B=T((1,1,0),$ $(2,1,0),$ $(1.5,0.5,1),$ $(1.5,1.5,1))$ and let $B'=T((-2,1,0),$ $(-1,1,0),$ $(-1.5,0.5,1),$ $(-1.5,1.5,1))$.

We first show that the following two cases can not occur:
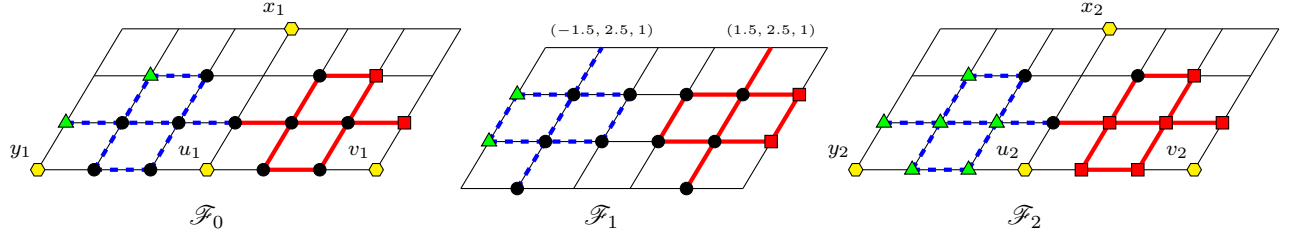
\begin{figure}[t]
\begin{center}
\begin{tikzpicture}[scale=0.75]

\draw (0-1,2.5/3-0.5) -- (5,2.5/3-0.5);
\draw (0-0.5,5/3-0.5) -- (6-0.5,5/3-0.5);
\draw (0.5-0.5,2) -- (6,2);
\draw (0.5,2+2.5/3) -- (6.5,2+2.5/3);
\draw (-1,2.5/3-0.5) -- (0.5,2+2.5/3);
\draw (0,2.5/3-0.5) -- (1.5,2+2.5/3);
\draw (1,2.5/3-0.5) -- (2.5,2+2.5/3);
\draw (2,2.5/3-0.5) -- (3.5,2+2.5/3);
\draw (3,2.5/3-0.5) -- (4.5,2+2.5/3);
\draw (4,2.5/3-0.5) -- (5.5,2+2.5/3);
\draw (5,2.5/3-0.5) -- (6.5,2+2.5/3);

\draw (0+6.5,0) -- (5+6.5,0);
\draw (0.5+6.5,2.5/3) -- (5.5+6.5,2.5/3);
\draw (1+6.5,5/3) -- (6+6.5,5/3);
\draw (1.5+6.5,2.5) -- (6.5+6.5,2.5);
\draw (0+6.5,0) -- (1.5+6.5,2.5);
\draw (1+6.5,0) -- (2.5+6.5,2.5);
\draw (2+6.5,0) -- (3.5+6.5,2.5);
\draw (3+6.5,0) -- (4.5+6.5,2.5);
\draw (4+6.5,0) -- (5.5+6.5,2.5);
\draw (5+6.5,0) -- (6.5+6.5,2.5);

\draw (0-1+14.5,2.5/3-0.5) -- (5+14.5,2.5/3-0.5);
\draw (0-0.5+14.5,5/3-0.5) -- (6-0.5+14.5,5/3-0.5);
\draw (0.5-0.5+14.5,2) -- (6+14.5,2);
\draw (0.5+14.5,2+2.5/3) -- (6.5+14.5,2+2.5/3);
\draw (-1+14.5,2.5/3-0.5) -- (0.5+14.5,2+2.5/3);
\draw (0+14.5,2.5/3-0.5) -- (1.5+14.5,2+2.5/3);
\draw (1+14.5,2.5/3-0.5) -- (2.5+14.5,2+2.5/3);
\draw (2+14.5,2.5/3-0.5) -- (3.5+14.5,2+2.5/3);
\draw (3+14.5,2.5/3-0.5) -- (4.5+14.5,2+2.5/3);
\draw (4+14.5,2.5/3-0.5) -- (5.5+14.5,2+2.5/3);
\draw (5+14.5,2.5/3-0.5) -- (6.5+14.5,2+2.5/3);

\draw[ultra thick, color=red] (0+3,2.5/3-0.5) -- (1+3,2);
\draw[ultra thick, color=red] (1+3,2.5/3-0.5) -- (2+3,2);
\draw[ultra thick, color=red] (1-1+3,2.5/3-0.5) -- (1+3,2.5/3-0.5);
\draw[ultra thick, color=red] (-0.5+3,5/3-0.5) -- (3-0.5+3,5/3-0.5);
\draw[ultra thick, color=red] (1.5-0.5+3,2) -- (2+3,2);
\draw[ultra thick, color=red]  (0.5+6.5+3,2.5/3) -- (1+6.5+3,5/3);
\draw[ultra thick, color=red]  (1+6.5+3,0) -- (2.5+6.5+3,2.5);
\draw[ultra thick, color=red]  (2.5+6.5+3,2.5/3) -- (3+6.5+3,5/3);
\draw[ultra thick, color=red]  (1.5/3+6.5+3,2.5/3) -- (2+1.5/3+6.5+3,2.5/3);
\draw[ultra thick, color=red]  (1+6.5+3,5/3) -- (3+6.5+3,5/3);
\draw[ultra thick, color=red] (0+14.5+3,2.5/3-0.5) -- (1+14.5+3,2);
\draw[ultra thick, color=red] (1+14.5+3,2.5/3-0.5) -- (2+14.5+3,2);
\draw[ultra thick, color=red] (1-1+14.5+3,2.5/3-0.5) -- (1+14.5+3,2.5/3-0.5);
\draw[ultra thick, color=red] (-0.5+14.5+3,5/3-0.5) -- (3-0.5+14.5+3,5/3-0.5);
\draw[ultra thick, color=red] (1.5-0.5+14.5+3,2) -- (2+14.5+3,2);

\draw[ultra thick, color=blue, dashed] (0,2.5/3-0.5) -- (1,2);
\draw[ultra thick, color=blue, dashed] (1,2.5/3-0.5) -- (2,2);
\draw[ultra thick, color=blue, dashed] (1-1,2.5/3-0.5) -- (1,2.5/3-0.5);
\draw[ultra thick, color=blue, dashed] (-0.5,5/3-0.5) -- (3-0.5,5/3-0.5);
\draw[ultra thick, color=blue, dashed] (1.5-0.5,2) -- (2,2);
\draw[ultra thick, color=blue, dashed]  (0.5+6.5,2.5/3) -- (1+6.5,5/3);
\draw[ultra thick, color=blue, dashed]  (1+6.5,0) -- (2.5+6.5,2.5);
\draw[ultra thick, color=blue, dashed]  (2.5+6.5,2.5/3) -- (3+6.5,5/3);
\draw[ultra thick, color=blue, dashed]  (1.5/3+6.5,2.5/3) -- (2+1.5/3+6.5,2.5/3);
\draw[ultra thick, color=blue, dashed]  (1+6.5,5/3) -- (3+6.5,5/3);
\draw[ultra thick, color=blue, dashed] (0+14.5,2.5/3-0.5) -- (1+14.5,2);
\draw[ultra thick, color=blue, dashed] (1+14.5,2.5/3-0.5) -- (2+14.5,2);
\draw[ultra thick, color=blue, dashed] (1-1+14.5,2.5/3-0.5) -- (1+14.5,2.5/3-0.5);
\draw[ultra thick, color=blue, dashed] (-0.5+14.5,5/3-0.5) -- (3-0.5+14.5,5/3-0.5);
\draw[ultra thick, color=blue, dashed] (1.5-0.5+14.5,2) -- (2+14.5,2);

\node at (0-1,2.5/3-0.5)[regular polygon, regular polygon sides=6,draw=black,fill=yellow,scale=0.5]{};
\node at (3-1,2.5/3-0.5)[regular polygon, regular polygon sides=6,draw=black,fill=yellow,scale=0.5]{};
\node at (6-1,2.5/3-0.5)[regular polygon, regular polygon sides=6,draw=black,fill=yellow,scale=0.5]{};
\node at (0,2.5/3-0.5)[circle,draw=black,fill=black,scale=0.5]{};
\node at (1,2.5/3-0.5)[circle,draw=black,fill=black,scale=0.5]{};
\node at (3,2.5/3-0.5)[circle,draw=black,fill=black,scale=0.5]{};
\node at (4,2.5/3-0.5)[circle,draw=black,fill=black,scale=0.5]{};
\node at (-0.5,5/3-0.5)[regular polygon, regular polygon sides=3,draw=black,fill=green,scale=0.35]{};
\node at (0.5,5/3-0.5)[circle,draw=black,fill=black,scale=0.5]{};
\node at (1.5,5/3-0.5)[circle,draw=black,fill=black,scale=0.5]{};
\node at (2.5,5/3-0.5) [circle,draw=black,fill=black,scale=0.5]{};
\node at (3.5,5/3-0.5)[circle,draw=black,fill=black,scale=0.5]{};
\node at (4.5,5/3-0.5) [circle,draw=black,fill=black,scale=0.5]{};
\node at (5.5,5/3-0.5)[regular polygon, regular polygon sides=4,draw=black,fill=red,scale=0.5]{};
\node at (1,2)[regular polygon, regular polygon sides=3,draw=black,fill=green,scale=0.35]{};
\node at (2,2)[circle,draw=black,fill=black,scale=0.5]{};
\node at (4,2)[circle,draw=black,fill=black,scale=0.5]{};
\node at (5,2)[regular polygon, regular polygon sides=4,draw=black,fill=red,scale=0.5]{};

\node at (3+6.5,5/3)[circle,draw=black,fill=black,scale=0.5]{};
\node at (2.5+6.5,2.5/3)[circle,draw=black,fill=black,scale=0.5]{};
\node at (3+6.5-1,5/3)[circle,draw=black,fill=black,scale=0.5]{};
\node at (2.5+6.5-1,2.5/3)[circle,draw=black,fill=black,scale=0.5]{};
\node at (3+6.5-2,5/3)[regular polygon, regular polygon sides=3,draw=black,fill=green,scale=0.35]{};
\node at (4.5+6.5-4,2.5/3)[regular polygon, regular polygon sides=3,draw=black,fill=green,scale=0.35]{};
\node at (1+6.5,0)[circle,draw=black,fill=black,scale=0.5]{};
\node at (3+9.5,5/3)[regular polygon, regular polygon sides=4,draw=black,fill=red,scale=0.5]{};
\node at (2.5+9.5,2.5/3)[regular polygon, regular polygon sides=4,draw=black,fill=red,scale=0.5]{};
\node at (3+9.5-1,5/3)[circle,draw=black,fill=black,scale=0.5]{};
\node at (2.5+9.5-1,2.5/3)[circle,draw=black,fill=black,scale=0.5]{};
\node at (3+9.5-2,5/3)[circle,draw=black,fill=black,scale=0.5]{};
\node at (4.5+9.5-4,2.5/3)[circle,draw=black,fill=black,scale=0.5]{};
\node at (1+9.5,0)[circle,draw=black,fill=black,scale=0.5]{};

\node at (0-1+14.5,2.5/3-0.5)[regular polygon, regular polygon sides=6,draw=black,fill=yellow,scale=0.5]{};
\node at (3-1+14.5,2.5/3-0.5)[regular polygon, regular polygon sides=6,draw=black,fill=yellow,scale=0.5]{};
\node at (6-1+14.5,2.5/3-0.5)[regular polygon, regular polygon sides=6,draw=black,fill=yellow,scale=0.5]{};
\node at (0+14.5,2.5/3-0.5)[regular polygon, regular polygon sides=3,draw=black,fill=green,scale=0.35]{};
\node at (1+14.5,2.5/3-0.5)[regular polygon, regular polygon sides=3,draw=black,fill=green,scale=0.35]{};
\node at (3+14.5,2.5/3-0.5)[regular polygon, regular polygon sides=4,draw=black,fill=red,scale=0.5]{};
\node at (4+14.5,2.5/3-0.5)[regular polygon, regular polygon sides=4,draw=black,fill=red,scale=0.5]{};
\node at (-0.5+14.5,5/3-0.5)[regular polygon, regular polygon sides=3,draw=black,fill=green,scale=0.35]{};
\node at (0.5+14.5,5/3-0.5)[regular polygon, regular polygon sides=3,draw=black,fill=green,scale=0.35]{};
\node at (1.5+14.5,5/3-0.5)[regular polygon, regular polygon sides=3,draw=black,fill=green,scale=0.35]{};
\node at (2.5+14.5,5/3-0.5) [circle,draw=black,fill=black,scale=0.5]{};
\node at (3.5+14.5,5/3-0.5)[regular polygon, regular polygon sides=4,draw=black,fill=red,scale=0.5]{};
\node at (4.5+14.5,5/3-0.5) [regular polygon, regular polygon sides=4,draw=black,fill=red,scale=0.5]{};
\node at (5.5+14.5,5/3-0.5)[regular polygon, regular polygon sides=4,draw=black,fill=red,scale=0.5]{};
\node at (1+14.5,2)[regular polygon, regular polygon sides=3,draw=black,fill=green,scale=0.35]{};
\node at (2+14.5,2)[circle,draw=black,fill=black,scale=0.5]{};
\node at (4+14.5,2)[circle,draw=black,fill=black,scale=0.5]{};
\node at (5+14.5,2)[regular polygon, regular polygon sides=4,draw=black,fill=red,scale=0.5]{};
\node at (3.5,2+2.5/3)[regular polygon, regular polygon sides=6,draw=black,fill=yellow,scale=0.5]{};
\node at (3.5+14.5,2+2.5/3)[regular polygon, regular polygon sides=6,draw=black,fill=yellow,scale=0.5]{};

\node at (0.9-2.2,0.65){\footnotesize{$y_{1}$}};
\node at (3.9-2.2,0.65){\footnotesize{$u_{1}$}};
\node at (6.9-2.2,0.65){\footnotesize{$v_{1}$}};
\node at (3.2,2+2.5/3+0.35){\footnotesize{$x_{1}$}};

\node at (0.9-2.2+14.5,0.65){\footnotesize{$y_{2}$}};
\node at (3.9-2.2+14.5,0.65){\footnotesize{$u_{2}$}};
\node at (6.9-2.2+14.5,0.65){\footnotesize{$v_{2}$}};
\node at (3.2+14.5,2+2.5/3+0.35){\footnotesize{$x_{2}$}};

\node at (2.5+6.5,2.8){\tiny{$(-1.5,2.5,1)$}};
\node at (5.5+6.5,2.8){\tiny{$(1.5,2.5,1)$}};

\node at (2,-0.5){$\mathcal{F}_{0}$};
\node at (16.5,-0.5){$\mathcal{F}_{2}$};
\node at (9,-0.5){$\mathcal{F}_{1}$};

\end{tikzpicture}
\end{center}
\caption{The vertices in $B\cup B'\cup \{(-3,1,2),$ $(0,1,2),$ $(3,1,2)\} $ at distance at most $3$ from $u_1$ (circles), $v_2$ (squares) or $y_2$ (triangles) and the vertices $u_1$, $u_2$, $v_1$, $v_2$, $x_1$, $x_2$, $y_1$ and $y_2$ (hexagons; thick lines: edges between vertices in $B\cup \{(0,1,2), (3,1,2)\}$; dashed lines: edges between vertices in $B'\cup \{(-3,1,2),(0,1,2)\}$).}
\label{twoheadedshark}
\end{figure}

\begin{description}
\item[Case 1: $c(u_1)=c(v_2)=c(x_2)$ or $c(u_2)=c(v_1)=c(x_1)$.]
By Lemma~\ref{le2667}, since $c((0,1,0))\neq c((3,1,2))$ or $c((0,1,2))\neq c((3,1,0)$, one vertex among $(0,1,2)$ and $(3,1,2)$ has a color not given to the vertices of $B$ (since  $(0,1,0)$ is at distance at most $3$ of every vertex of $B\setminus \{(3,1,2)\}$ and that $(3,1,0)$ is at distance at most $3$ of every vertex of $B\setminus \{(0,1,2)\}$).
Suppose $c(u_1)=c(v_2)=c(x_2)$.
It can be easily checked that every vertex of $B\cup\{ (0,1,2),$ $(3,1,2)\}$ is at distance at most $3$ of (at least) one vertex among $\{u_1,v_2,x_2\}$.
Figure \ref{twoheadedshark} illustrates the fact that every vertex of $B\setminus\{ (1.5,2.5,1)\}\cup\{ (0,1,2),$ $(3,1,2)\}$ is at distance at most $3$ of one vertex among $\{u_1,v_2\}$. It can be easily verified that $d_{\mathcal{F}_{0,2}}(x_2,(1.5,2.5,1))\le 3$.

Thus, $c(u_1)$ is different from the $23$ used colors. Therefore, we obtain a contradiction.

In the case $c(u_2)=c(v_1)=c(x_1)$ it can be remarked that every vertex of $B\cup\{ (0,1,2),$ $(3,1,2)\}$ is at distance at most $3$ of (at least) one vertex among $\{ u_2, v_1,x_1\}$ and, as previously, we obtain a contradiction since $B\cup\{ (0,1,2),$ $(3,1,2)\}$ contain vertices of 23 different colors.

\item[Case 2: $c(u_1)=c(v_2)=c(y_2)$ or $c(u_2)=c(v_1)=c(y_1)$.]

Suppose that $c(u_1)=c(v_2)=c(y_2)$.
It can be easily verified that every vertex of $(B\setminus\{(1.5,2.5,1)\})\cup \{(0,1,2),$ $(3,1,2)\}$ is at distance at most $3$ of (at least) one vertex among $\{u_1,$ $v_2\}$ and that every vertex of $(B'\setminus\{(-1.5,2.5,1)\})\cup \{(0,1,2),$ $(-3,1,2)\}$ is at distance at most $3$ of one vertex among $\{u_1,$ $y_1 \}$.
In Figure \ref{twoheadedshark}, we illustrates the fact that every vertex of $(B\setminus\{(1.5,2.5,1)\})\cup \{(0,1,2),$ $(3,1,2)\}$ is at distance at most $3$ of one vertex among $\{u_1,$ $v_2\}$.
By Lemma~\ref{le2667}, one vertex among $(0,1,2)$ and $(3,1,2)$ has a color not given to the vertices of $B$ and one vertex among $(0,1,2)$ and $(-3,1,2)$ has a color not given to the vertices of $B'$.
Since both $B\cup\{ (0,1,2),$ $(3,1,2)\}$ and $B'\cup \{(0,1,2),$ $(-3,1,2)\}$ contain vertices of 23 different colors, we have $c(u_1)=c((1.5,2.5,1))$ and $c(u_1)=c((-1.5,2.5,1))$. However, since $d_{\mathcal{F}_{0,2}}((1.5,2.5,1),(-1.5,2.5,1))=3$, we have a contradiction.

In the case $c(u_2)=c(v_1)=c(y_1)$, since every vertex of $(B\setminus\{(1.5,2.5,1)\})\cup \{(0,1,2),$ $(3,1,2)\}$ is at distance at most $3$ of one vertex among $\{u_2,$ $v_1\}$ and every vertex of $(B'\setminus\{(-1.5,2.5,1)\})\cup \{(0,1,2),$ $(-3,1,2)\}$ is at distance at most $3$ of one vertex among $\{u_2,$ $y_1\}$, we have also a contradiction since, as previously, both $B\cup\{ (0,1,2),$ $(3,1,2)\}$ and $B'\cup \{(0,1,2),$ $(-3,1,2)\}$ contain vertices of 23 different colors (this implies $c(u_2)=c((1.5,2.5,1))$ and $c(u_2)=c((-1.5,2.5,1))$.
\end{description}

By Lemma \ref{le2661}, either $c(u_1)=c(v_2)$ or $c(u_2)=c(v_1)$ holds.
First, suppose $c(u_1)=c(v_2)$. By Case 1, we have $c(u_1)\neq c(x_2)$ and by Case 2, we have $c(u_1)\neq c(y_2)$. Consequently, by Lemma \ref{le2661}, we have $c(u_2)=c(x_1)=c(y_1)$ and by symmetry and Lemma \ref{le2661}, we can not have both $c(u_1)\neq c(y_2)$ and $c(u_2)\neq c(y_1)$ and a contradiction with Case 1.
Second suppose $c(u_2)=c(v_1)$. By Case 1, we have $c(u_2)\neq c(x_1)$ and by Case 2, we have $c(u_2)\neq c(y_1)$. Consequently, by Lemma \ref{le2661}, we have $c(u_1)=c(x_2)=c(y_2)$, and a contradiction with Case 1.

Finally, we obtain that $\mathcal{F}_{0,2}^3$ is not $23$-colorable.
\end{proof}
\begin{cor}
We have $\chi(\mathcal{F}_{0,2}^3)= 24$.
\end{cor}
\section*{Acknowledgments}
This work was supported by the French "Investissements d'Avenir" program, project ISITE-BFC (contract ANR-15-IDEX-03).


\begin{thebibliography}{99}
\bibitem{BJO1990}
S.~Bjornholm, Clusters, condensed matter in embryonic form, 
{\em Contemp. Phys.} \textbf{31}, 309--324, 1990.
\bibitem{BOUR2015}
J.~Bourgeois, S.~Copen Goldstein,
Distributed intelligent MEMS: progresses and perspectives, {\em IEEE Systems Journal} \textbf{9} (3), 1057--1068, 2015.
\bibitem{JJD2017}
J.~J.~Daymude, R.~Gmyr, A.~W.~Richa, C.~Scheideler, T.~Strothmann,
Improved leader election for self-organizing programmable matter, {\em ALGOSENSORS 2017: Algorithms for Sensor Systems}, 127--140, 2017.
\bibitem{FE2003}
G.~Fertin, E.~Godard, A.~Raspaud, Acyclic and $k$-distance coloring of the grid, {\em Inform. Process. Lett.} \textbf{87} (1)51--58, 2003.
\bibitem{JJ2005}
P.~Jacko, S.~Jendrol', Distance coloring of the hexagonal lattice, {\em Discuss. Math. Graph Theory} \textbf{25}, 151--166, 2005.
\bibitem{KR2008}
F.~Kramer, H.~Kramer,
A survey on the distance-colouring of graphs,
 {\em Discrete Math.} \textbf{308}, 422--426, 2008.
\bibitem{BOUR2018}
B.~Piranda, J.~Bourgeois,
Geometrical Study of a Quasi-spherical Module for Building Programmable Matter, {\em Distributed Autonomous Robotic Systems}, 347--400, 2018.
\bibitem{SE2001}
A.~Sevcikova, Distant chromatic number of the planar graphs, {\em Manuscript}, Safarik University, 2001.
\bibitem{SP2016}
P. Šparl, R. Witkowski, J. Žerovnik, Multicoloring of cannonball graphs, {\em Ars Mathematica Contemporanea}  \textbf{10} (1), 31--44, 2016.
\end{thebibliography}
\end{document}